\documentclass[journal]{IEEEtran}
\usepackage[T1]{fontenc}% optional T1 font encoding
\usepackage{amsthm}
\usepackage{float}
\usepackage{caption}
\usepackage{subcaption}
\usepackage{algorithm}
\usepackage{algorithmic}
\usepackage[cmintegrals]{newtxmath}
\usepackage{amsmath}
\usepackage{color}
\usepackage{tikz}
\usepackage{url}
\usepackage{lipsum,adjustbox}
\def\BState{\State\hskip-\ALG@thistlm}
\newtheorem{theorem}{Theorem}
\newtheorem{proposition}{Proposition}
\newtheorem{remark}{Remark}
\newtheorem{lemma}{Lemma}
\newtheorem{corollary}{Corollary}
\newtheorem{definition}{Definition}
\newlength\myindent
\newcommand{\norm}[1]{\left\lVert#1\right\rVert}
\newcommand\bindent{%
  \begingroup
  \setlength{\itemindent}{\myindent}
  \addtolength{\algorithmicindent}{\myindent}
}
\newcommand\eindent{\endgroup}
\newcommand\myeq{\mathrel{\overset{\makebox[0pt]{\mbox{\normalfont\tiny\sffamily c}}}{\equiv}}}
 \usepackage[compact]{titlesec}
    \titlespacing{\section}{0pt}{2ex}{1ex}
    \titlespacing{\subsection}{0pt}{1ex}{0ex}
    \titlespacing{\subsubsection}{0pt}{0.5ex}{0ex}

\newenvironment{tight}
    {
    \vspace{\abovedisplayskip}
    \setlength{\abovedisplayskip}{0pt}
    \setlength{\abovedisplayshortskip}{0pt}
    \setlength{\belowdisplayskip}{0pt}
    \setlength{\belowdisplayshortskip}{0pt}
    }{
    \vspace{10.0pt plus 2.0pt minus 5.0pt}
    }

\begin{document}

\title{Secure Aggregation in Federated Learning \\ using Multiparty Homomorphic Encryption}

% \author{Erfan Hosseini, Shuangyi Chen, Ashish Khisti% <-this % stops a space
% }
\author{%
  \IEEEauthorblockN{Erfan Hosseini\IEEEauthorrefmark{1},
                    Shuangyi Chen\IEEEauthorrefmark{1},
                    and Ashish Khisti\IEEEauthorrefmark{1}}
                    
  \IEEEauthorblockA{\IEEEauthorrefmark{1}%
                    University of Toronto, \\
                    \{ehosseini2108@gmail.com, shuangyi.chen@mail.utoronto.ca, akhisti@ece.utoronto.ca\}}
}

\maketitle

\begin{abstract}
A key operation in federated learning is the aggregation of gradient vectors generated by individual client nodes. We develop a method based on  multiparty homomorphic encryption (MPHE) that enables the central node to compute this aggregate, while receiving only encrypted version of each individual gradients. Towards this end, we extend classical MPHE methods so that the decryption of the aggregate vector can be successful even when only a subset of client nodes are available. This is accomplished by introducing a secret-sharing step during the setup phase of MPHE when the public encryption key is generated. We develop conditions on the parameters of the MPHE scheme that guarantee correctness of decryption and (computational) security.  We explain how our method can be extended to accommodate client nodes that do not participate during the setup phase. We also propose a compression scheme for gradient vectors at each client node that can be readily combined with our MPHE scheme and perform the associated convergence analysis. We discuss the advantages of our proposed scheme with other approaches based on secure multi-party computation. 
Finally we discuss a practical implementation of our system, compare the performance of our system with different approaches, and demonstrate that by suitably combining compression with encryption the overhead over baseline schemes is rather small.
\end{abstract}

% \begin{IEEEkeywords}
% Federated Learning, Homomorphic Encryption, Cryptography, Secure Multiparty Computation, Secure Aggregation, Compression, Secret Sharing
% \end{IEEEkeywords}

\IEEEpeerreviewmaketitle

\section{Introduction}

Federating learning (FL) is a framework for training machine learning models without transmitting local datasets\cite{DBLP:journals/corr/McMahanMRA16}. It does not require central storage and provides some level of privacy. In this framework, participants contribute to the learning of a global model by sharing their local model or local gradients instead of their data points. A central node aggregates the local updates and performs the global update. Even though, some level of privacy is achieved, the local model can still leak information about the local datasets. An eavesdropper can obtain information about the local datasets by observing the model parameters or gradients\cite{DBLP:conf/ccs/FredriksonJR15}\cite{DBLP:journals/corr/abs-1805-04049}. Secure aggregation is an approach in FL where users transmit encrypted version of the gradients, such that the central node is able to compute the aggregate but not infer any additional information about the individual gradients. 

Multiparty homomorphic encryption (MPHE) enables a group of distributed users, each having access to a private variable, to participate in jointly computing a function (e.g.,  aggregate sum) without leaking any additional information about each variable \cite{10.1007/978-3-642-29011-4_29}.  MPHE appears naturally suited for secure aggregation in FL applications, where users  encrypt their respective gradients and transmit to the  central node, which must facilitate in computing the aggregate. In this approach, users need to communicate with each other during an initial setup phase to generate a common public encryption key. Thereafter, when training the global model, the users  communicate with the central node their gradients which are encrypted using this public key. The decryption of the aggregate sum cannot be directly performed as the secret-key associated with this public key is not stored at any node. The decryption proceeds in a collaborative manner by transmitting shares generated at each user to the central node. 

 MPHE techniques however pose certain challenges in FL applications. First, the decryption step in MPHE requires the presence of all the users that participate in the generation of the public key. In FL applications this is not practical ---  the central node selects only a subset of users at each iteration of training. Secondly, MPHE schemes do not have provision for users to join who do not participate in the generation of the public key. Finally, in FL applications the gradient vectors are very high dimensional, consisting of several millions of parameters. Existing MPHE implementations  can be noticeably slow in these applications\cite{cryptoeprint:2017:956}\cite{cryptoeprint:2013:094}\cite{cryptoeprint:2015:345}\cite{10.1007/978-3-642-29011-4_29}. 
Motivated by these considerations, we propose a novel solution for MPHE that is suitable in practical FL applications. Our solution involves a setup phase where participants exchange shares of their respective secret-keys based on Shamir's secret-sharing. This approach guarantees that the presence of any $k$ users is sufficient for decryption, where $k$ is a design parameter. Secondly we discuss a simple extension of our approach that can enable users who were not present during the setup phase to participate and transmit their encrypted gradients. Finally we discuss a technique for compression of gradient vectors that provides noticeable speed up during the encryption and decryption steps. 

In related works, a number of different approaches have been developed for secure aggregation in FL. Secure multi-party computation (SMPC) is a popular method for achieving security. Although these methods achieve privacy, they usually require high communication complexity.  Reference \cite{10.1145/3133956.3133982} proposes a scheme for secure aggregation that can tolerate user dropouts by applying pairwise additive masking. However, this method has quadratic communication complexity. Prior works have proposed methods for reducing this complexity to $\mathcal{O}(N\log N)$ by reducing the connectivity \cite{choi2020communicationcomputation}. Furthermore, TurboAggregate  proposes a different approach to reduce the communication complexity to $O(N \log N)$ by creating multi-group circular aggregation structures\cite{DBLP:journals/corr/abs-2002-04156}.
Bell et al. \cite{10.1145/3372297.3417885} introduces the Harray graph for pairwise additive masking with Shamir's secret sharing, achieving $\mathcal{O}(N\log N)$ communication complexity but requiring new graphs and key pairs per iteration, unlike our one-time key generation approach. Homomorphic encryption is also widely used in FL, allowing clients to encrypt models and transmit them for aggregation by the central server. Some methods share the secret key among users to ensure secure collaboration while keeping it confidential from the server \cite{chen2024secure,Dong2020EaSTFLyEA,9664035,9148628,10.5555/3489146.3489179,10405290,jin2024fedmlheefficienthomomorphicencryptionbasedprivacypreserving,Ma_2022}. Zhang et al. \cite{10.5555/3489146.3489179} proposed BatchCrypt using quantization and HE, but it assumes no collusion between users and the server, a weak premise as the server can compromise users' model privacy. Ma et al. \cite{Ma_2022} proposed Multi-Key HE for secure aggregation but suffer from failed aggregation with client dropouts and require key regeneration and re-broadcasting for new clients. Our approach avoids both issues. Threshold HE systems \cite{10.5555/648118.746742, truex2019hybrid, 9355600, ma2021privacypreserving, tian2024lattice} overlook scenarios where new clients join after the setup phase. Another line of works \cite{10.1145/3338501.3357371, chen2023quadratic, 10064312} use Functional Encryption \cite{boneh2011functional}. HybridAlpha \cite{10.1145/3338501.3357371} employs multi-input Functional Encryption with a trusted third party (TTP), which may be unrealistic in practical setting. 

% Federated learning libraries provide support for secure aggregation through all three mentioned methods. The communication  and key management is handled through a trusted third party\cite{10.1145/3338501.3357371}. 

The rest of the paper is organized as follows: First, we give a background on secret sharing, homomorphic encryption and multiparty homomorphic encryption in Section~\ref{sec:background}. Then, we formally introduce the system model in Section~\ref{sec:statement}. We present proposed solution and provide correctness and security analysis for our scheme in Section~\ref{algorithm_definition}. Subsequently, we discuss the connection of our scheme to federated learning and propose a compression scheme to improve the communication and computation requirements in Section~\ref{sec:compression}. We then discuss an extension to allow for participation of new users. We provide experimental results to evaluate the performance of our framework in Section~\ref{sec:exper}.
% The proofs of some of the theorems in Section~\ref{sec:compression} are relegated to the Appendix. 

\section{Background}
\label{sec:background}

\subsection{Notation}
We denote $[\cdot]_q$ the reduction of an integer modulo $q$, and $\lceil \cdot \rceil$, $\lfloor \cdot \rfloor$, $\lfloor \cdot \rceil$ the rounding to the next, previous, and nearest integer respectively. When applied to polynomials, these operations are performed coefficient-wise. We use regular letters for integers and polynomials, and bold letters for vectors. Given a probability distribution $\mathcal{D}$ over a ring $R$, $p \longleftarrow \mathcal{D}$ denotes sampling an element according to $\mathcal{D}$, and $p \longleftarrow R$ implicitly denotes uniform sampling in $R$. For a polynomial $a$, we denote its infinity norm by $\norm{a}$. We use $ \mathbb{Z}_q = [-\frac{q}{2}, \frac{q}{2})$ as the representatives for the congruence classes modulo $q$. Table~\ref{tab:notation_summary} provides the summary of the symbols used throughout this work.
\begin{table}[h]
    \centering
    \begin{tabular}{c|c}
        Symbol & Description \\
        \hline
        $\longleftarrow$ & uniform sampling \\
        $\norm{\cdot}$ & infinity norm \\
        $\norm{\cdot}_1$ & l-$1$ norm \\
        $\norm{\cdot}_2$ & l-$2$ norm \\
        $\mathbb{Z}_q$ & $ [-\frac{q}{2}, \frac{q}{2})$ \\
        % $\chi$ & 
        % \begin{tabular}{c}
        %     centred discrete Gaussian with variance $\sigma^2$ \\ and truncated support over $[-B, B]$ 
        % \end{tabular}
        % \\
        $R_q$ & $\mathbb{Z}_q[X]/(X^n+1)$ \\
        
        $n$ & degree of the encryption polynomial \\ 
        $\lambda$ & security parameter \\
        $q$ & size of ciphertext space \\
        $p$ & size of plaintext space \\
        $N$ & number of clients \\
        $k$ & minimum available clients \\
        $d$ & dimensionality of input/gradient vector \\
        $s$ & dimesnionality of the compressed gradients \\
        $r$ & compression ratio \\
        $T$ & number of aggregation rounds \\
        $m$ & message \\
        $p_1$ & public random polynomial \\
        $sk$ & secret key \\
        $\boldsymbol{cpk}$ & common public key \\
        $\boldsymbol{pk}$ & public key \\
        $\boldsymbol{ct}$ & ciphertext \\
        ${\mathcal P}_t = \{a_i\}$ & set of available users \\
        $\boldsymbol{g}_i^t$ & gradient/input vector \\
        $\bar{\boldsymbol{g}}^t$ & aggregated gradient/input vector \\ 
        $A_{s, \chi}^q$ & decision-RLWE distribution \\ 
        $\delta$ & compression coefficient \\ 
        $\mathbf{\Phi}$ & compression matrix \\

    \end{tabular}
    \caption{Notation Summary}
    \label{tab:notation_summary}
\end{table}
\subsection{Shamir's Secret Sharing}
\label{shamir_scheme}
First we define $(N, k)$ secret sharing scheme. \cite{10.1145/359168.359176}
\begin{definition}
    Let $s \in \mathbb{Z}_p$ be the secret, where $p$ is a prime. We aim to divide the secret into $N$ chunks, such that 
    \begin{itemize}
        \item Knowledge of any $k$ chunks allow reconstruction of $s$
        \item Knowledge of any $k-1$ or fewer chunks would provide no information about $s$ 
    \end{itemize}
    A given algorithm with the mentioned properties is called a $(N, k)$ secret sharing scheme. 
\end{definition}

The algorithm works by constructing a random polynomial $f(x)$ of order $k-1$ so that $f(0) = s$ and $f(x) = s + \sum_{j=1}^{k-1} b_j x^j$, where $b_j \gets \mathbb{Z}_p$ are selected uniformly random. The shares are generated by distributing the evaluations of $f(\cdot)$ over a set of distinct evaluation points denoted by $x_j \in \mathbb{Z}_p$. We can reconstruct the polynomial $f(\cdot)$ and recover the secret using $k$ evaluation points . The set $\{(x_i, s_i)\}_{i=0}^{N-1} $ denotes secret shares of $s$. 

We discuss an extension of Shamir's secret sharing where the secret $s$ is a polynomial with coefficients in $\mathbb{Z}_p$. The set of secret shares $(x_i, s_i)$ are constructed according to the following
\begin{equation}
    \label{shamir_construction}
    s_i = s + \sum_{j=1}^{k-1} x_i ^ j t_{j},
\end{equation}
where $\{t_j\}_{j=1}^{k-1}$ is a set of $k-1$ polynomials of the same order as $s$ and coefficients uniformly sampled from $\mathbb{Z}_p$. The mentioned approach applies secret sharing to each coefficient individually using the same evaluation points. The following well-known result shows that given any set of $k$ secret shares, we are able to reconstruct the random coefficients.
\begin{theorem}{\cite{10.1145/359168.359176}}
    \label{shamir_reconstruction}
    Given any set of $k$ shares, $\{(x_{a_i}, s_{a_i})\}_{i=0}^{k-1}$ generated according to~\eqref{shamir_construction}, one can reconstruct the secret $s$ i.e., there exist a set of coefficients $\{r_i \}_{i=0}^{k-1}, r_i \in \mathbb{Z}_p$ such that
    \begin{equation}
        s = \sum_{i=0}^{k-1} s_{a_i} \cdot r_i.
        \label{eq:linear}
    \end{equation}
    In addition the polynomials $t_j$ in~\eqref{shamir_construction} can also be constructed using these shares, i.e., the set $\{r_{i, j}\}_{i=0}^{k-1},  r_{i, j} \in \mathbb{Z}_p$ exists so that $t_j = \sum_{i=0}^{k-1} s_{a_i} r_{i,j}$.
\end{theorem}

The following shows the additive property of Shamir's secret sharing scheme following linearity in~\eqref{eq:linear}.
\begin{remark}
    \label{additivie_property_shamir}
    Let $\{x_i, s_{1, i} \}$ and $\{x_i, s_{2, i} \}$ be the set of secret shares for $s_1$ and $s_2$ generated according to $(N,k)$ Shamir's scheme. The set defined as $\{x_i, s_{1, i}+s_{2, i} \}$ are secret shares of $s_1 + s_2$.
\end{remark}
Furthermore, it is known that any collection of $k-1$ secret shares are mutually independent. 
\begin{theorem}
    \label{independence_shamir_shares}
    Let $\mathcal{S} = \{s_{a_0}, ..., s_{a_{k-2}}\}$ denote an arbitrary set of $k-1$ secret shares generated according to $(N,k)$ Shamir's scheme. The shares are mutually independent and distributed uniformly over $Z_p$. 
\iffalse
    \begin{proof}
    Each share generated according to \ref{shamir_construction} is a weighted sum of uniformly distributed random variables. Therefore, each share is itself uniformly distributed over $\mathbb{Z}_p$ given the corresponding evaluation point. 
    Now we prove that any set of $k$ shares are jointly uniform over $\mathbb{Z}_p^k$.  We have $p^k$ distinct possibilities for $\mathcal{S}$ each representing a unique equiprobable polynomial with an order of $k-1$ and coefficients in $\mathbb{Z}_p$. Since we have a total of $p^k$ polynomials with an order of $k-1$, we conclude that the the set of secret shares are jointly uniform.
    \end{proof}
\fi
\end{theorem}

\subsection{BFV Encryption Scheme}
\label{BFV}
The Brakerski-Fan-Vercauteren cryptosystem \cite{cryptoeprint:2012:144} is a ring learning-with-errors scheme,  that supports both additive and multiplicative homomorphic operations. The ciphertext space is $R_q = \mathbb{Z}_q[X]/(X^n+1)$, the quotient ring of the polynomials with coefficients in $\mathbb{Z}_q$ modulo $(X^n+1)$, where $n$ is a power of 2. Unless otherwise stated, we consider the arithmetic in $R_q$, so polynomial reductions are omitted in the notation. The plaintext space is the ring $R_p = \mathbb{Z}_p[X]/(X^n+1)$ for $p<q$. We denote $\Delta = \lfloor \frac{q}{p} \rfloor$. The scheme is based on two distributions: (i) the key distribution over $R_3$ with uniformly distributed coefficients;  (ii) the RLWE error distribution $\chi$ over $R_q$ has coefficients distributed according to a centred discrete Gaussian with variance $\sigma^2$ and truncated support over $[-B, B]$. 

We next discuss the operation of the BFV cryptosystem as our subsequent methods will rely on it. 

\bigskip
\noindent{\em Key-Generation Functions}:
\begin{itemize}
    \item $BFV.SecKeyGen()$: Sample $s \longleftarrow R_3$ and output the secret key $sk=s$
    \item $BFV.PubKeyGen(s)$: Sample $p_1 \longleftarrow R_q$, and $e_{\boldsymbol{pk}} \longleftarrow \chi$  and output the public key: \begin{equation}
    \boldsymbol{pk} = (p_0, p_1) = (-(sp_1+e_{\boldsymbol{pk}}), p_1).\label{eq:pk}
    \end{equation}
 \end{itemize}
 
 \bigskip
 \noindent{ \em Encryption and Decryption Functions}:
  Since we allow the input  vector ${\mathbf g}$ to be in the space $\mathbb{Z}^d$ while the plaintext is in the ring $R_t$, therefore we define a function $m = Encode({\mathbf g})$ and a decoding function ${\mathbf g} = Decode(m)$ in our description below.
 
 \begin{itemize}
    \item $BFV.Encrypt(\boldsymbol{pk}, \boldsymbol{g})$: Let $\boldsymbol{pk} = (p_0, p_1)$. Sample $u \longleftarrow R_3$ and $e_0, e_1 \longleftarrow \chi$. Output the ciphertext
    \begin{equation}
        \boldsymbol{ct} = (c_0, c_1) =  (\Delta \cdot m + up_0 + e_0, u p_1 + e_1)         \label{BFV_encryption}
    \end{equation}
  where recall that  $\Delta = \lfloor \frac{q}{p} \rfloor$ as introduced earlier.
    \item $BFV.Decrypt(s, \boldsymbol{ct})$: Let $ct = (c_0, c_1)$ and compute $m' = [\lfloor \frac{p}{q} [c_0 + c_1 s] \rceil]_p$.

\end{itemize}
Upon substituting~\eqref{eq:pk} and~\eqref{BFV_encryption} can show that:
\begin{equation}
    \label{eq:decrypt}
    c_0 + sc_1 = \Delta \cdot m + e_{\boldsymbol{ct}}, 
\end{equation}
where
\begin{equation}
    \label{eq:ect}
    e_{\boldsymbol{ct}} = e_0 + se_1 - ue_{\boldsymbol{pk}}  
\end{equation}
denotes the error in the ciphertext. From the properties of the BFV cryptosystem it can be established that~\cite{cryptoeprint:2012:144} the decryption is correct i.e., we can guarantee that $m' =m$, as long as this error stays small, $\norm{e_{\boldsymbol{ct}}} < \frac{q}{2p}$. We have 
\begin{align}
    \label{BFV_correctness}
    \norm{e_{\boldsymbol{ct}}} &\leq \norm{e_0 + se_1 - ue_{\boldsymbol{pk}}} \leq \norm{e_0} + \norm{se_1} +  \norm{ue_{\boldsymbol{pk}}} \nonumber \\
    &\leq B + n\norm{s}\norm{e_1} + n\norm{u}\norm{e_{\boldsymbol{pk}}} \leq B(2n+1).
\end{align}
Therefore, the ciphertext decryption is achievable if $B(2n+1)<\frac{q}{2p}$.
We can map the plaintext $m$ into the input vector ${\mathbf g}$ using the $Decode(\cdot)$ function.

We note that the security of the BFV scheme relies on the following property \cite{10.1145/2535925}.
\begin{definition}
    \label{RLWE_definition}
    (Hardness of Decision-RLWE) For a security parameter $\lambda$, let $q=q(\lambda)>2$ be an integer. For a random $s \in R_q$ (the secret) and a noise distribution $\chi = \chi(\lambda)$ over $R_q$. Denote by $A_{s, \chi}^q$ the distribution obtained by choosing a uniformly random $a \gets R_q$ and a noise term $e \gets \chi$ and outputting $(a, as+e)$. The hardness of RLWE problem states that it is computationally hard for a polynomial time adversary to distinguish between $A_{s, \chi}^q$ and a uniform distribution over $R_q^2$.\cite{10.1145/2535925}
    
 \end{definition}
    \begin{remark}
    We denote the security parameter by $\lambda$. Parameters of the system are chosen based on the desired security level. Larger order polynomials, $n$, and noise variance, $B$, would result in a more secure scheme. However, a larger ciphertext space, $q$, decreases the security of the network \cite{cryptoeprint:2012:144}.   
\end{remark}

The above structure in the BFV scheme immediately yields the following additive property,

\begin{proposition}
\label{additivi_property_BFV}
    Let us consider a BFV scheme with a secret-key $s$ and the associated public-key $\boldsymbol{pk}$,
    \begin{equation}
        \boldsymbol{ct}_1 = (ct_{1,0}, ct_{1,1}) = BFV.Encrypt(pk, \boldsymbol{g_1}),
    \end{equation}
    \begin{equation}
        \boldsymbol{ct}_2 = (ct_{2,0}, ct_{2,1}) = BFV.Encrypt(pk, \boldsymbol{g_2}).
    \end{equation}
    Let $\boldsymbol{ct} = (ct_{1,0} + ct_{2,0}, ct_{2,1} + ct_{2,1})$ be the addition of the two ciphertexts. Furthermore, let $e_{\boldsymbol{ct}_1}$ and $e_{\boldsymbol{ct}_2}$ denote the noise of the given ciphertexts defined as \eqref{eq:decrypt}
    \begin{equation}
        BFV.Decrypt(s, \boldsymbol{ct}) = \boldsymbol{g_1} + \boldsymbol{g_2},
    \end{equation}
    if $\norm{e_{\boldsymbol{ct}_1} + e_{\boldsymbol{ct}_2}} < \frac{q}{2p}$.
    
\end{proposition}

Note that Prop.~\ref{additivi_property_BFV} indicates that if two inputs ${\mathbf g_1}$ and ${\mathbf g_2}$ are encrypted separately using the public-key $pk$, and we add the associated ciphertexts, then 
the decryption of this quantity using the secret-key $s$ will result in the sum ${\mathbf g_1} + {\mathbf g_2}$.
\subsection{Multiparty Homomorphic Encryption}
\label{MPHE}
Multi-party homomorphic encryption (MPHE) is a technique where $N$ users, each with a private input, can collectively compute a function over these inputs (e.g., an aggregate of the inputs) without leaking any additional information of their private inputs. To motivate MPHE, consider a setting with $N=3$ users, where user $i$ has access to an input ${\mathbf g_i}$, which must be encrypted and transmitted to a central node such that only the sum ${\mathbf g_1}+{\mathbf g_2} + {\mathbf g_3}$ is only known publicly.
The approach suggested by Prop.~\ref{additivi_property_BFV} cannot be directly used. If the decryption key $s$ is revealed to the central node, it can in fact decrypt the ciphertexts $\boldsymbol{ct}_i$ to directly obtain ${\mathbf g_i}$. MPHE enables an interesting application of applying homomorphic encryption, where the secret key $s$ is not revealed at any node but only certain shares of $s$ are made available to each node.
 
The algorithm consists of two phases as discussed below \cite{Mouchet2020MultipartyHE}:

\subsubsection{Setup Phase} In this phase, users agree on a collective public key. This phase provides each client with a common public key and a secret share of the collective secret key. The clients are assumed to have access to a synchronized source of randomness that is used during setup phase.

The generation of the public key proceeds as follows (for the node $i$):
\begin{enumerate}
    \item Sample common random polynomial $p_1 \longleftarrow R_3$. This is known to all the participating nodes.  
    \item Sample a secret key $s_i \longleftarrow R_3$ (at each node $i$).
    \item Sample $e_i \longleftarrow \chi$ and disclose $p_{0,i} = -(p_1s_i +e_i)$.
    \item Compute $p_0 = \sum_{j=0}^{N-1} p_{0,j}$.
\end{enumerate}
The collective public key $\boldsymbol{cpk}=(p_0, p_1)$ is the output of this phase.
\begin{proposition}
    The output of the setup step, $\boldsymbol{cpk}$, is a public key corresponding to  $s = \sum_i s_i  \label{eq:sum}$
\end{proposition}
\begin{proof}
This claim follows immediately from construction above by substituting the expression for $p_{0,i}$ into the expression for $p_0$.
\end{proof}
\subsubsection{Decryption}
In this phase, we assume that clients have access to a publicly known ciphertext $\boldsymbol{ct}$, and wish to compute the encrypted message $\boldsymbol{g}$. The ciphertext $\boldsymbol{ct}$ is given by\footnote{In practice $\boldsymbol{ct}$ is generated in a distributed fashion which is described in section \ref{algorithm_definition}. In our proposed example we can consider $\boldsymbol{ct} = \boldsymbol{ct}_1+ \boldsymbol{ct}_2 + \boldsymbol{ct}_3$ as discussed before.} $\boldsymbol{ct} = BFV.Encrypt(\boldsymbol{cpk} , \boldsymbol{g})$. 
We wish to apply the decryption operation associated with the secret-key $s$ to compute ${\mathbf g} = BFV.Decrypt(s, {\mathbf{ct}})$. Naturally we do not have access to $s$ at any node in the network as $s$ is given by~\eqref{eq:sum}. 
According to the BFV decryption scheme (see~\eqref{eq:decrypt}), given $\boldsymbol{ct} = (c_0, c_1)$, we wish to compute $c_0 + s c_1$ using~\eqref{eq:sum}. Note that since, $c_0 + s c_1 = c_0 + \sum_i s_i c_1$ it would be sufficient for each user to transmit $s_i c_1$ to make decryption possible. However, this would leak $s_i$ as $c_1$ is publicly known. Therefore, an additional noise term is added to protect $s_i$. This noise term is called smudging noise and is sampled randomly from $\chi_{smg}$, a uniform distribution over $[-B_{smg}, B_{smg}]$. 

This phase proceeds as following for user $i$
\begin{enumerate}
    \item Samples $e_i \gets \chi_{smg}$.
    \item disclose $h_i = s_i c_1 + e_i$.
\end{enumerate}
The output of this phase is $m' = [\lfloor \frac{p}{q} [c_0 + \sum_{i} h_{i} ] \rceil]_p$. The message can then be recovered by $Decode(m')$.

We discuss the correctness and security of the above scheme under more general conditions in Section~\ref{section_analysis}.

\section{System Model}
\label{sec:statement}
We focus on a setup with a central node and $N$ clients. We assume that client $i$ has a sequence of input vectors $\{\boldsymbol{g}_i^t\}_{t=0}^{t=T-1}$. At iteration $t$, some subset of $k$ users denoted by the set ${\mathcal P}_t$, are available. We aim to compute the aggregate of the input vectors denoted by $\Bar{\boldsymbol{g}}^t = \sum_{i \in {\mathcal P}_t} \boldsymbol{g}_{i}^t$. 
We wish to compute $\{ \Bar{\boldsymbol{g}}^t \}_{t=0}^{T-1}$ in a privacy-preserving manner meaning that no information should be leaked about the individual inputs.  

\iffalse
We allow the users to have a setup phase. During setup, we assume all the users are responsive. Let $B_{set}$ be the communication complexity for the setup phase. We should note that even if the mentioned quantity is large, it has negligible effect with regards to the overall period $T$ such that $\lim_{T \rightarrow \infty} \frac{B_{set}}{T+1} = 0$.
For the rest of the protocol, $t=0, ..., T-1$, we assume that at least $k$ clients are available.
\fi 

\paragraph{Threat Model and Assumptions}
We assume all the parties to be honest-but-curious meaning they correctly follows the algorithm and protocol but may try to learn private inputs from their observations. We also assume there can be a limited number of dishonest parties who may try to infer honest parties' private information. Dishonest parties may collude with each other and share observations to infer private inputs of honest parties. In our protocol, the number of dishonest parties is constrained to be at most $k-1$ out of a total of $N+1$ parties. More specifically, there can be $k-1$ colluded clients acting dishonestly or a combination of $k-2$ clients and the central server colluded together.
% We assume the central server to be honest-but-curious meaning it correctly follows the algorithm and protocol but may try to learn private inputs from its observations. 

% With respect to the clients parties, we assume there can be a limited number of dishonest parties who may try to infer honest parties' private information. Dishonest parties may collude with each other and share observations to infer private inputs of honest parties. In our protocol, the number of dishonest parties is bounded by $k$ out of $N$ parties. Additionally, we assume the central node and clients do not collude.

We assume that the system allows for a setup phase when all the clients are responsive and a training phase where at-least $k$ clients are responsive.

\paragraph{Privacy Goals}
Our main objective is to enable secure aggregation of clients' gradient vectors in the above system and the threat model. We aim to protect each client's gradient vector confidentiality, that is, during iterations, the server should not learn any honest client's gradient vector and client $i$ should not learn the gradient vector of client $j$ if $j \neq i$.

\paragraph{Baseline}

Note that under the special case of $k=N$, i.e., the presence of all users is guaranteed at each iteration, the method mentioned in Section \ref{MPHE} can be used. In this example, users run the key generation step to generate $\{s_i \}$ and $\boldsymbol{cpk}$. Then at each iteration, users encrypt their inputs and transmit it to the central node, $\boldsymbol{ct}_i^t = BFV.Encrypt(\boldsymbol{cpk}, \boldsymbol{g}_i^t)$. The central node performs aggregation over the encrypted values $\boldsymbol{ct}^t = \sum_i \boldsymbol{ct}^t_i$. The result of aggregation is the encrypted aggregated gradient under the same key such that $BFV.Decrypt(s, \boldsymbol{ct}) = \sum_i \boldsymbol{g}^t_i$, where $s=\sum s_i$ is the common secret key corresponding to the $\boldsymbol{cpk}$. 
The users then proceed with the decryption stage to recover the aggregated gradients. The issue with the above protocol is that it requires $k=N$ i.e., all users must be present during the decryption phase and hence the entire training phase. In practice in federated learning applications only a subset of users are selected at each iteration and the requirement that $k=N$ cannot be satisfied in practice. 

 {\em Adaptive Secure Aggregation {\bf ASA}}: As another baseline, one way of approaching the user selection problem in the above extension is by running the entire process at each iteration for the set of available users. At any given iteration, the central node selects a subset of $k$ users among the available set. The users perform the setup phase followed by the decryption phase in each iteration. The scheme is regarded as adaptive secure aggregation (ASA) throughout this work.
 We will see in the experiments that the cost of having to run a setup phase at each iteration is higher than our proposed schemes. Furthermore the scheme assumes that client dropouts cannot occur between the setup phase and encryption phase in each iteration, which might not hold in practice.

\section{Robust Secure Aggregation (RSA)}
\label{algorithm_definition}
Our proposed protocol is an extension of the MPHE scheme in Section~\ref{MPHE} to allow the participation of any $k$ out of $N$ users in the decryption phase.
It consists of two phases:
\begin{enumerate}
    \item \textbf{Setup phase} 
     aims to produce a collective public key and secret shares of the collective secret key. Secret shares are generated in order to make decryption possible without the need to reveal the underlying secret key. 
    \item \textbf{Aggregation} step performs the aggregation over the encrypted values and then performs collaborative decryption by extending the approach in Section~\ref{MPHE}.
\end{enumerate}
The setup phase is only performed once at the beginning of the process. However, the aggregation step can be run multiple times.

\subsection{Setup Phase}
\label{setup_phase_our}
\begin{algorithm}[h]
\caption{Setup phase}
\label{key_generation_algorithm}
\begin{algorithmic}
    \STATE \textbf{Inputs}: encryption space $R_q$, noise distribution $\chi$, common random polynomial $p_1$, set of evaluation points $\{x_i\}_{i=0...N-1}$
    \STATE \textbf{Client $i$:} 
    \bindent
    \STATE $s_i \gets R_3$ 
    \STATE $t_{i, l} \gets R_q$ for $l=1 ... k-1$ 
    \STATE $s_{i,j} = s_i + \sum_{l=1}^{k-1} t_{i,l} x_j^l$
    \STATE securely transmit $s_{i, j}$ to client $j$
    \STATE $s'_i = \sum_j s_{j,i}$
    \STATE sample $e_i \gets \chi$ 
    \STATE transmit $p_{0,i} = -(p_1 s_i + e_i)$ to the central node
    \eindent
    \STATE \textbf{Central node:} 
    \bindent
    \STATE $p_0 = \sum_i p_{0, i}$
    \STATE output $\boldsymbol{cpk}=(p_0, p_1)$
    \eindent
\end{algorithmic}
\end{algorithm}

This phase aims to generate (i) a key pair, common secret key, and public key, denoted by $(s, \boldsymbol{cpk})$ and (ii) a secret share of the secret-key at each node $s'_i$.

The common public key  $\boldsymbol{cpk}$ is available to all of the participants of the network whereas the associated secret key $s = \sum_i s_i$ is not known to any individual. Until this point, the execution of the setup phase is the same as in Section~\ref{MPHE}. We note that the original MPHE algorithm requires access to common randomness. We relax this assumption by generating the common random polynomial in the central node and broadcasting it to the users. 

The additional component is that the knowledge of the common secret key is distributed in the network using a $(N, k)$ Shamir's secret sharing scheme. Thus instead of simply having access to $s_i$, each node has access to a share $s'_i$ using an evaluation point $x_i$ such that the collection of any $k$ shares suffices to reconstruct the secret-key $s$. This property intuitively allows for the decryption to proceed even when only $k$ nodes are present. We note that the set of evaluation points are generated and transmitted to the users by the central node.

\begin{remark}
    We assume that all communication takes place through the central node. In order to facilitate private communication among users through a public channel, we use a public key encryption system similar to \cite{10.1145/3133956.3133982} \cite{DBLP:journals/corr/abs-2002-04156}.
    %We denote this operation as secure transmission in our descriptions.
\end{remark}

To generate the necessary shares, the main steps of client $i$ are as follows:
\begin{enumerate}
    \item Sample $s_i \longleftarrow R_3$.
    \item Generate secret shares of $s_i$ denoted by the set $\{(x_j, s_{i,j}) \}_{j=1}^N$ as per the following expression  $s_{i,j} = s_i + \sum_{l=1}^{l=k-1} t_{i,l} x_j ^ l$,
    where $\{ t_{i,l} \}$ are polynomials with coefficients uniformly sampled from $\mathbb{Z}_q$, see~\eqref{shamir_construction}. 
    \item Securely transmit $s_{i,j}$ to client $j$ through the central node.
    \item Aggregate received shares 
    \begin{equation}
        \label{secret_shares_setup}
        s'_i = \sum_j s_{j,i}.    
    \end{equation}
    
\end{enumerate}
%All the secret shares in the network are generated according to the same evaluation points.

\begin{proposition}
    \label{secret_shares_addition}
    The quantities $\{s'_i \}$ computed in step (5) above  in the key-generation phase are secret shares of $s=\sum_i s_i
    $.
    \begin{proof}
    Let $f_i(x)$ denote the random polynomial used by user $i$ to generate secret shares $f_i(x) = s_i + \sum_{l=1}^{l=k-1} t_{i,l} x ^ l$. We note that shares are generated by evaluating $f_i(x)$ at corresponding points: $s_{i,j} = f_i(x_j)$. Let $f(x)$ denote the global random polynomial constructed by summing the local random polynomials $f(x) = \sum_{i=0}^{N-1} f_i(x)$.
    We show that $\{s'_i\}_{i=0}^{N-1}$ is evaluations of $f(x)$ at $\{x_i\}_{i=0}^{N-1}$ 
    \begin{equation}
        s'_i = \sum_{j=0}^{N-1} s_{j,i} = \sum_{j=0}^{N-1} f_j(x_i) = f(x_i),
    \end{equation}
    which follows from Proposition~\ref{secret_shares_addition} and Theorem~\ref{independence_shamir_shares}.
    \end{proof}
\end{proposition}
\begin{corollary}
    \label{secret_shares_independence}
    Any set of $k$  secret shares $s'_i$ generated according to expression~\eqref{secret_shares_setup} are mutually independent and uniformly distributed over $R_q$.
\end{corollary}
%This is a direct consequence of remark \ref{additivie_property_shamir}.

\begin{remark}
Note that the secret-key $s$ has the following desired properties: (i) it is not known to any individual and  (ii) a collective knowledge required to utilize it is present among any $k$ of the clients.
\end{remark}

\begin{remark}
The above process requires communication between each pair of nodes. Thus the communication complexity scales as $N^2$, where $N$ is the number of nodes in the network. However, this phase is run only once and is independent of $T$.

\end{remark}

\subsection{Aggregation Phase}
\begin{algorithm}[h]
\caption{Aggregation}
\label{aggregation_algorithm}
\begin{algorithmic}
    \STATE \textbf{Inputs}: common public key $\boldsymbol{cpk}$, input vectors $\{ \boldsymbol{g_i^t} \}$, smudging noise distribution $\chi_{smg}$, number of rounds $T$
    \FOR{t<T}
    \STATE \textbf{Client $i$:} 
    \bindent
    \STATE $\boldsymbol{ct}_i^t = BFV.Encrypt(\boldsymbol{cpk}, \boldsymbol{g}^t_i)$
    \STATE transmit $\boldsymbol{ct}^t_i$ to the central node
    \eindent
        \STATE \textbf{Central node:} 
    \bindent
        \STATE $\boldsymbol{ct}^t=(c_0^t, c_1^t) = \sum_i \boldsymbol{ct}^t_i$
        \STATE broadcast $\boldsymbol{ct}^t$
        \STATE select $k$ active users $\{ a_i \}_{i=0}^{k-1}$
        \STATE transmit reconstruction coefficients $\{ r_{a_i}^t \}_{i=0}^{k-1}$
    \eindent
    \STATE \textbf{Client $a_i$} (any selected in iteration $t$):
    \bindent
        \STATE $e_{a_i}^t \gets \chi_{smg}$
        \STATE $h_{a_i}^t = r_{a_i}^t \cdot s'_{a_i} \cdot c_1^t + e_{a_i}^t$
        \STATE transmit $h_{a_i}^t$ to the central node
    \eindent
        \STATE \textbf{Central node:}
    \bindent
        \STATE $m' = [\lfloor \frac{p}{q} [c_0^t + \sum_{a_j} h_{a_j}^t ] \rceil]_p$
        \STATE $\bar{g}^t = Decode(m')$
    \eindent
    \ENDFOR
\end{algorithmic}
\end{algorithm}

At this phase, users compute and encrypt their inputs. Then a subset of the active users is selected by the central node to participate in the decryption. Selected users generate decryption shares which are aggregated by the central node to recover the message.
\subsubsection{Encryption}
At this step we assume that client $i$ computes its gradient $\boldsymbol{g}^t_i$ and encrypts it into $\boldsymbol{ct}^t_i$
\begin{equation}
    \label{eq:encrypt}
    \boldsymbol{ct}^t_i = BFV.Encrypt(\boldsymbol{cpk}, \boldsymbol{g}^t_i).
\end{equation}
Clients transmit their ciphertexts to the central node for aggregation. 
\begin{equation}
    \boldsymbol{ct}^t = \sum_i \boldsymbol{ct}^t_i,
\end{equation}
where the sum happens element-wise. 
The aggregation result, $\boldsymbol{ct}^t$, is then sent back to the users for decryption. 
%We should note that we do not have any participation constraints in this step meaning that we can proceed if at least one user submits its encrypted gradient.
\subsubsection{User Selection}
At this step, each client transmits their availability to the central node, non-responsive nodes are marked as inactive after a time-out period. Then, $k$ active users are arbitrarily selected, denoted by $\{a_i\}_{i=0}^{k-1}$. The central node computes and transmits the reconstruction coefficient of the selected users according to the scheme discussed at section \ref{shamir_scheme} denoted by $\{ r_{a_i}^t \}_{i=0}^{k-1}$. We know that the coefficients satisfy the following property:
\begin{equation}
    \sum_{i=0}^{k-1} s'_{a_i} r_{a_i}^t = s. \label{eq:shares}
\end{equation}
% \begin{remark}
% {\color{red}
%     Our protocol provides the flexibility to choose a new set of $k$ clients for aggregation in each training iteration. Notably, during the aggregation phase, the clients who contribute their model updates and the clients responsible for distributed decryption can be distinct groups.}
% \end{remark}
\subsubsection{Decryption}
\label{decryption_step}
This steps aims to decrypt $\boldsymbol{ct}^t$ to produce $\sum_i {\mathbf g_i^t}.$ Following, 
Prop.~\ref{additivi_property_BFV}, since $\boldsymbol{ct}^t$ encrypts the aggregated gradient under $\boldsymbol{cpk}$, we have that:
\begin{equation}
    \sum_i \boldsymbol{g}^t_i = BFV.Decrypt(s, \boldsymbol{ct}^t).
\end{equation}
We note that $s$ is not accessible to any individual in the network, therefore, $\boldsymbol{ct}^t$ is not decryptable in the current form. Instead, users generate decryption shares by decrypting $\boldsymbol{ct}^t$ using their local shares scaled by their reconstruction coefficient. We should note that we require to compute the following quantity to decrypt assuming that $\boldsymbol{ct}^t = (c_0^t, c_1^t)$
\begin{equation}
    c_0^t + s \cdot c_1^t =     c_0^t + \sum_{i=0}^{k-1} r_{a_i}^t s'_{a_i} c_1^t,
\end{equation}
where we have used~\eqref{eq:shares}.
As a result, it is sufficient for each selected client $a_j$ to only transmit $r_{a_j}^t s'_{a_j} c_1^t$. Following Section~\ref{MPHE}, additional noise is superimposed in order to ensure the security.
The process is as follows:
\begin{enumerate}
    \item Samples $e_{a_i}^t \gets \chi_{smg}$.
    \item Transmits $h_{a_i}^t = r_{a_i}^t s'_{a_i} c_1^t  + e_{a_i}^t$ to the central node.
\end{enumerate}
Central node then generates 
\begin{equation}
    \label{final_pt_recovery}
    m' = [\lfloor \frac{p}{q} [c_0^t + \sum_{a_j} h_{a_j}^t ] \rceil]_p.
\end{equation}
Finally the aggregated gradients are available through $Decode(m')$ as discussed before.
\begin{theorem}
\label{thm:correctness}
Let $m$ and $e_{\boldsymbol{ct}^t}$ be the message and ciphertext noise~\eqref{eq:ect} resulting from decryption of $\boldsymbol{ct}^t=(c_0^t, c_1^t)$ under secret key $s$ such that $c_0^t + sc_1^t = \Delta \cdot m + e_{\boldsymbol{ct}^t}$. Let $e_{smg}=\sum_{a_i} e_{a_i}^t$ be the sum of all smudging noise introduced in step (3) and $m'$ be the recovered plaintext according to~\eqref{final_pt_recovery}. The recovered plaintext is consistent such that $m'=m$ if the noise stays within the decryptable bounds such that $\norm{e_{\boldsymbol{ct}^t} + e_{smg}} < \frac{q}{2p}$.
\begin{proof}
    According to~\eqref{final_pt_recovery} we have 
    \begin{align}
        \label{decryption_proof}
        m' &= [\lfloor \frac{p}{q} [c_0^t + \sum_{a_i} h_{a_i}^t]\rceil]_p 
        = [\lfloor \frac{p}{q} [c_0^t + \sum_{a_i} r_{a_i}^t s'_{a_i}c_1^t + e_{a_i}^t]\rceil]_p \nonumber \\  
        &= [\lfloor \frac{p}{q} [c_0^t + sc_1^t + e_{smg}]\rceil]_p
        = [\lfloor \frac{p}{q} [\Delta \cdot m + e_{\boldsymbol{ct}^t} + e_{smg}]\rceil]_p.
\end{align}
Following the discussion in Section~\ref{BFV}, we have $m'=m$, if $\norm{e_{\boldsymbol{ct}^t} + e_{smg}} < \frac{q}{2p}$.
\end{proof}

\end{theorem}

\begin{remark}
    Our protocol provides the flexibility to choose a new set of $k$ clients for aggregation in each training iteration. Notably, during the aggregation phase, the clients who contribute their model updates and the clients responsible for decryption can be distinct groups.
\end{remark}

\section{Analysis}
\label{section_analysis}
We  provide conditions that guarantee correctness of the algorithm. Then we discuss conditions for  security.
\subsection{Correctness Analysis}
In this section we provide bounds on the noise induced by the multiparty homomorphic encryption setup. As a result of the setup phase, each secret key $s_i$ is sampled from $R_3$, we know that $\norm{s} \leq N$. Let $\boldsymbol{ct_i}=(c_0, c_1)$ be a fresh ciphertext encrypting an input vector under collective public key $\boldsymbol{cpk}$ at user $i$. Decryption under the collective secret key outputs $c_0 + sc_1 = \Delta \cdot m + e_{\boldsymbol{ct}_i}$. Following~\eqref{BFV_correctness} and considering $\norm{s} \leq N$, we have $\norm{e_{\boldsymbol{ct}_i}} \leq B(2nN + 1)$. Therefore, the ciphertext noise scales linearly in $N$. 

Following the aggregation step of the algorithm we know $\boldsymbol{ct} = \sum \boldsymbol{ct}_i$. Therefore,
\begin{equation}
    \label{noise_bound_aggregate_ciphertext}
    \norm{e_{\boldsymbol{ct}}} \leq BN(2nN + 1).
\end{equation}

During decryption we have $\norm{e_{smg}} \leq kB_{smg}$. 
\begin{theorem} (Correctness)
    Let $N, k, q, t, n, B, B_{smg}$ be the system parameters representing number of clients, robustness threshold, size of ciphertext space, size of plaintext space,  order of encryption polynomial, bound on BFV noise and bound on the smudging noise. The algorithm correctly recovers the aggregated inputs if $BN(2nN+1)+kB_{smg} < \frac{q}{2p}$.
    \proof Immediately follows  Theorem~\ref{thm:correctness} and expression~\eqref{noise_bound_aggregate_ciphertext}.
\end{theorem}
\subsection{Security Analysis}
Throughout, we let $\lambda$ denote the security parameter and $\mu(\lambda)$ denote a negligible function. 
\begin{definition}
    (Negligible Function) We say a function $\mu:\mathbb{N} \rightarrow \mathbb{R}$ is negligible if for any given polynomial $p(\cdot)$, there exists an integer $\rho_0$ such that $\mu(\rho) < \frac{1}{p(\rho)}$ for all $\rho > \rho_0$.
\end{definition}
\begin{definition}
Given a security parameter $\lambda$, we say two distribution ensembles $X$ and $Y$ are computationally indistinguishable, denoted by $X \myeq{Y}$, if for every polynomial time algorithm $D$ there exists negligible function $\mu(\cdot)$ such that $|Pr[D(X) = 1] - Pr[D(Y)=1]| < \mu(\lambda)$.
\end{definition}
\begin{lemma}
    \label{smudging_lemma}
    Let $B_1$ and $B_2$ be two positive integers. Let $e_1 \in [-B_1, B_1]$ be a fixed integer and $e_2 \gets [-B_2, B_2]$ selected uniformly at random. The distribution of $e_2$ is computationally indistinguishable from that of $e_2 + e_1$ if $\frac{B_1}{B_2} < \mu(\lambda)$.
    \proof Refer to \cite{10.1007/978-3-642-29011-4_29}.
\end{lemma}
We further analyze the security of the scheme in the passive semi-honest adversary model. We provide arguments based on the simulation formalism\cite{Lindell2017} defined as below.

\begin{proposition} %proposition and add the definition of view
(\textit{Simulation Formalism}) Let $y = f(x_0, ..., x_{N-1})$ be a deterministic functionality. We say a protocol securely realizes $f$ in the presence of a polynomial time static semi-honest adversary $\mathcal{A}$ if a simulator $S$ exists so that $S( \{x_i\}_{i \in \mathcal{A}}, y) \myeq V_{\mathcal{A}}$, where $V_{\mathcal{A}}$ is the real view of $\mathcal{A}$ representing the collection  of private inputs, messages received during the execution of the protocol and the internal state of variables.
\end{proposition} 
Let $\mathcal{P}$ denote all the participants in the algorithm. For any possible adversary $\mathcal{A}$, defined as a subset of at most $k-1$ participants, we construct a simulator program $S(\cdot)$ so that given the inputs and outputs of the adversary, can simulate $\mathcal{A}$'s view in the algorithm. For a given value $x$, we denote the simulated equivalent by $\Tilde{x}$. We examine the security of the scheme by analyzing setup phase and decryption step.

Let $\mathcal{P}_a$ denote the set of active users at each stage. At any point of the algorithm, there exists at least one honest player denoted by $P_z$. Let $\mathcal{H}= \mathcal{P}_a \setminus ( \mathcal{A} \cup \{P_z\})$ denote the rest of the honest participants. Now we provide the security arguments for both phases of the algorithm.
\subsubsection{Setup Phase}
We consider an adversary $\mathcal{A}$ attacking the setup phase. We consider $s_i$ and $e_i$ as the private inputs for node $i$. Recall that these are used by node $i$ to generate $p_{0,i} = -(p_1 s_i + e_i)$ as in Section~\ref{key_generation_algorithm}. The output $p_0 = \sum_i p_{0,i}$.
Thus, the ideal functionality of this step is defined as 
\begin{equation}
    f_{SET} (p_1, \{ s_i, \{t_{i,l}\}_{l=1}^{k-1} ,e_i \}_{i=0}^{N-1}) = \left\{p_0, (s'_0, ..., s'_{N-1})\right\},
\end{equation}
where each output $s'_i$ is private to client $i$. 
\begin{theorem}
    \label{security_setup}
    The protocol introduced in section \ref{setup_phase_our} securely realizes $f_{SET}$ in the presence of a polynomial time static semi-honest adversary controlling up to $k-1$ parties.
    \begin{proof}
    Without loss of generality, we assume the adversary controls parties $\{0, ..., k-1\}$. We define the real view of $\mathcal{A}$ by including the common public polynomial $p_1$, the key generation shares $(p_{0, i})_{i \in {\mathcal P}}$, the secret shares generated by adversary nodes $(s_{i,j})_{i=0...k-1, j=0...N-1}$ and the secret shares received by the adversary nodes $(s_{i,j})_{i=k...N-1, j=0...k-1}$. We introduce the simulator $S(\cdot)$ as follows:  %$S\big(p_0, p_1, \{s_i, e_i\}_{i=0}^{k-1}, \{s'_i, t_{i,l}\}_{i=0,l=1}^{k-1}\big)$ as follows \textcolor{blue}{(This should have all $s'_i$?)}
    \begin{itemize}
    \item For adversary nodes the simulated values are set to be equal to the real values. For $i \in \mathcal{A}$ we set $\Tilde{p}_{0, i} = p_{0, i}$ and $\Tilde{s}_{i,j} = s_{i,j}$ for $j \in P$.
    \item For $i \in \mathcal{H}$ we set $\Tilde{p}_{0, i} \gets R_q$ and $\Tilde{s}_{i, j} \gets R_q$ for $j \in \mathcal{A}$.
    \item For $P_z$, we set $\Tilde{p}_{0, z} = p_0 - \sum_{P_j \in \mathcal{A} \cup \mathcal{H}} \Tilde{p}_{0, j}$, $\Tilde{s}_{z, j} = s'_j - \sum_{i \in \mathcal{A} \cup \mathcal{H}} \Tilde{s}_{i, j}$ for $j \in \mathcal{A} $.
    \end{itemize}
    Now we show that the simulated view is computationally indistinguishable from the real view. We note that given $p_1$, all the $p_{0,i}$ are mutually independent due to the independence of $s_i$ and $e_i$.
    Furthermore following Theorem \ref{independence_shamir_shares}  the secret-shares in the view are also independent. Therefore, it is sufficient to show the indistinguishability of each element separately. 
    
    Based on the hardness of decisional-RLWE problem (see Definition~\ref{RLWE_definition}), the simulated key generation shares are indistinguishable from the real ones such that $p_{0, i} \myeq \Tilde{p}_{0, i}$ for $i \in \mathcal{H}$. For the share of player $P_z$, we have
    \begin{align}
        \label{independence_sum}
        \Tilde{p}_{0, z} &= p_0 - \sum_{P_j \in \mathcal{A}}  p_{0, j} - \sum_{P_j \in \mathcal{H}}  \Tilde{p}_{0, j} \nonumber \\
        &\myeq  p_0 - \sum_{P_j \in \mathcal{A}}  p_{0, j} - \sum_{P_j \in \mathcal{H}}  p_{0, j} 
         \myeq p_{0, z}.
    \end{align}
  The relation $\Tilde{s}_{i,j} \myeq s_{i, j}$ directly follows Theorem \ref{independence_shamir_shares} for $i \in \mathcal{H}$. Furthermore, $\Tilde{s}_{z, j} \myeq s_{z, j}$ follows the same argument as ~\eqref{independence_sum}. This completes the proof. 
  \end{proof}
\end{theorem}

\subsubsection{Aggregation}
In this section we provide the security argument for the decryption step of the proposed scheme. Initially, we discuss one iteration of the algorithm and drop the dependence on the iteration index $t$. Assume that users hold ciphertext $\boldsymbol{ct}= (c_0, c_1)$ that is encrypting plaintext $m$ as we have defined in Eq.~\eqref{BFV_encryption}.
  
We consider $s'_i$ and $e_i$ as the private inputs of each party and denote the ideal functionality of this step as following
\begin{equation}
    f_{DEC}(\boldsymbol{ct}, \{ s'_i, e_i \}_{i \in {\mathcal P}_a}) = m.
\end{equation}
\begin{lemma}
    \label{RLWE_composition_lemma}
    Let $s \in R_q$ be a random element and $\chi$ a noise distribution. Furthermore let $u \gets R_q$ be a polynomial selected uniformly at random. Let $\Tilde{A}_{s, \chi}^q$ be the distribution defined by choosing $a \in R_q$ to be a random polynomial such that $a \myeq u$ and outputting $(a, as+e)$. We have $\Tilde{A}_{s, \chi}^q \myeq U(R_q^2)$, where $e \gets \chi$ is a randomly sampled from the noise distribution.
    \begin{proof}
        Based on Definition \ref{RLWE_definition}, $(a, as+e) \myeq (u, us+e) \myeq U(R_q^2)$.
    \end{proof}
\end{lemma}
\begin{theorem}
    \label{security_decryption}
    Given any fixed iteration $t$, the protocol proposed at section \ref{decryption_step} securely realizes $f_{DEC}$ in the presence of a polynomial time static semi-honest adversary controlling up to $k-1$ parties as long as $\frac{BN(2nN+1)}{kB_{smg}} < \mu(\lambda)$.
\end{theorem}
\begin{proof}
    Assume that $\mathcal{P}_t = \left\{a_0, \ldots a_{k-1}\right\}$ denotes the set of active users.
    We define the real view of the adversary as $c_1$ and $(h_{a_0}, ..., h_{a_{k-1}})$ (we do not mention the superscript as we only consider one iteration), the decryption shares generated by the users, recall that $h_{a_i} = r_{a_i} s'_{a_i} c_1 + e_{a_i}$ and $h = \sum_j h_{a_j}$ satisfies $c_0 + h = \Delta \cdot m + e_{\boldsymbol{ct}} + e_{smg}$,
    where $e_{\boldsymbol{ct}}$ is the decryption noise defined as \eqref{eq:ect} and $e_{smg}$ is the summation of all smudging noise. 
    We construct $S(\boldsymbol{ct}, \{ s'_i, e_i \}_{i \in \mathcal{A}},m)$ to simulate the real view. The simulated shares need to produce a consistent output. Therefore, the simulator begins by generating $\Tilde{h} = \Delta \cdot m + \Tilde{e}_{smg} - c_0$, where $\Tilde{e}_{smg}$ is distributed according to $\mathcal{D}$, a uniform distribution over $[-kB_{smg}, kB_{smg}]$. The simulated shares are set to be equal to the real shares for adversary nodes. Furthermore, for $P_i \in \mathcal{H}$, the simulated shares are selected uniformly at random $\Tilde{h}_{i} \gets R_q$. The share of player $P_z$ is generated as $\Tilde{h}_{z} = \Tilde{h} - \sum_{P_j \in \mathcal{A} \cup \mathcal{H}} \Tilde{h}_{j}$. 
    
    To prove that the simulated view is computationally indistinguishable from the real view, note that due to independence of noise terms $e_{a_i}$ as well as the result of Corollary~\ref{secret_shares_independence}, the variables $h_{a_i}$ are mutually independent given $c_1$. It is thus sufficient to show the computational indistinguishibility of elements of the view separately. 
    According to the smudging lemma (see \ref{smudging_lemma}) if $\frac{BN(2nN+1)}{kB_{smg}} < \mu(\lambda)$ we have $e_{\boldsymbol{ct}} + e_{smg} \myeq \Tilde{e}_{smg}$ 
    and therefore, $\Tilde{h} \myeq h$. 
    Let $u \gets R_q$ be a random polynomial selected uniformly at random. Based on Definition \ref{RLWE_definition} and~\eqref{BFV_encryption},  we know $c_1 \myeq u$ and thus for any fixed $r_{a_i}$, it follows that $r_{a_i}c_1 \myeq u$. Therefore, following Lemma~\ref{RLWE_composition_lemma}, we have $h_{a_i} \myeq u$. The case for $P_z$ follows as below:
    \begin{align}
        \Tilde{h}_z &= 
        \Tilde{h} - \sum_{P_j \in \mathcal{A} \cup \mathcal{H}} \Tilde{h}_{j} \myeq h - \sum_{P_j \in \mathcal{A} \cup \mathcal{H}} h_j \myeq h_z.
    \end{align}
    This completes the proof.
\end{proof}

Theorem \ref{security_decryption} provides the security argument for one iteration of the algorithm. Next theorem establishes the analysis for executing the aggregation for $T$ consecutive steps.
\begin{theorem}
    Given the conditions of theorem \ref{security_decryption}, a learning process consisting of $T$ iterations is computationally secure against a polynomial time static semi-honest adversary controlling a fixed set of users with size of at most $k-1$.
    \begin{proof}
        Let $\mathcal{P}_t$ denote the set of active users at iteration $t$. We define the real view of the adversary as $p_1$ and $(h_i^t)_{t=0...T-1, i \in \mathcal{P}_t}$ where $h_i^t = r_i^t s'_i c_1^t + e_i^t$. Theorem \ref{security_decryption} provides the security argument for each iteration independently. Let $P_z$ denote an arbitrary honest node. 
        Let $T_z \in \{0...T-1\}$ denote the time steps where the user is active in decryption. The view of the adversary regarding this user, during the execution of the learning process is defined as $\{h_z^t\}_{t \in T_z}$, where $h_z^t=s'_z r_z^t c_1^t + e_z^t$. We prove security by showing that this view is computationally indistinguishable from a uniform distribution. Given $p_1$, the elements of $\{r_z^t c_1^t\}_{t \in T_z}$ are mutually independent (see \eqref{BFV_encryption}). Furthermore, we know that $r_z^t c_1^t \myeq u$ where $u \gets R_q$ (refer to proof of Theorem \ref{security_decryption}). Therefore, $\{(r_z^t c_1^t, h_z^t)\}_{t \in T_z}$ is a set of $|T_z|$ samples generated according to $\Tilde{A}_{s'_z, \chi}^q$ (note that $s'_z$ is generated in the setup phase and fixed throughout) and based on Lemma \ref{RLWE_composition_lemma}, computationally indistinguishable from a uniform distribution. 

    \end{proof}
\end{theorem}

\subsection{Complexity Analysis}
\label{complexity_analysis}
The setup phase involves the transmission of secret shares. Each client creates $N$ secret shares of its local secret key and transmits it to its corresponding peer in the network. Therefore, the complexity is $O(N)$ per client. Additionally, the server receives the local public keys generated by the clients and therefore, the complexity is $O(N)$ for the server as well. As a result, the complexity of the entire network is $O(N^2)$. The aggregation phase involves three steps. At each step, first the server broadcasts the global weights to clients, then clients communicate a single message to the central node. Therefore, the complexity is $O(1)$ per client. The server receives at most $N$ messages from the users. Therefore, the complexity is $O(N)$ for the server. Table~\ref{tab:complexity_breakdown} demonstrates the complexity for different parts of the algorithm.

Assuming that the algorithm runs for a total of $T$ iterations, the complexity of the entire learning process would be $O(N^2 + NT)$. For cross-device applications, where $N >> T$, the dominating term in the complexity would be $O(N^2)$ and therefore, the complexity of entire process is quadratic in $N$. However, for cross-silo applications, where $T >> N$, the dominating term would be $O(NT)$. Therefore, our algorithm provides a solution with linear complexity. 
\begin{table}[h]
    \centering
    \begin{tabular}{c|c|c|c}
         & Client & Server & Network \\
         \hline
        Setup & $O(N)$ & $O(N)$ & $O(N^2)$ \\
        Aggregation & $O(1)$ & $O(N)$ & $O(N)$
    \end{tabular}
    \caption{Complexity of different parts of our algorithm}
    \label{tab:complexity_breakdown}
\end{table}

\section{Gradient Compression in Federated Learning}

\label{sec:compression}

In this section we discuss the application of our proposed method to federated learning, propose  a compression technique that is naturally compatible with our encryption technique and provide convergence guarantees of the training process associated with our proposed compression method. 

\subsection{Gradient Aggregation in Federated Learning}

We assume that there are $N$ client nodes, where client $i$ has dataset ${\mathcal S}_i$ for $i=0,1,\ldots, N-1$. The clients communicate with a central node who is interested in training a machine learning model using the datasets of the clients. Let ${\mathbf w}_t$ denote the parameters of this model at iteration $t$. We assume that ${\mathbf w}_0$ denote the values of the parameters at the initialization. The value of ${\mathbf w}_t$ is broadcast to the client nodes and client node $i$ computes the associated gradient vector update given by
\begin{equation}
\label{eq:grad}
{\mathbf g}_{i}^{t+1} = \nabla L\left({\mathcal S}_i, {\mathbf w}_t\right), \qquad t=0,\ldots, T-1,
\end{equation}
where $L\left({\mathcal S}_i, {\mathbf w}_t\right)$ denotes the loss function evaluated on the training set ${\mathcal S}_i$ when the model parameters are fixed to ${\mathbf w}_t$. These gradient updates ${\mathbf g}_{i}^{t+1}$ are transmitted to the central node, who then computes an aggregate\footnote{For simplicity, our rule assumes that the size of datasets ${\mathcal S}_i$ are the same for each node. If the size is different, a weighted sum must be computed. Our methods immediately extend to such setting. } $\bar{\mathbf g}^{t+1} = \frac{1}{N}\sum_{i=1}^N {\mathbf g}_{i}^{t+1}$. The aggregate gradient is then used to update the model parameters. If a simple gradient descent rule is used then we have
\begin{align}
\label{eq:update}
{\mathbf w}_{t+1} = {\mathbf w}_{t} - \gamma_t \cdot \bar{\mathbf g}^{t+1},
\end{align}
where $\gamma_t$ denotes the step size at iteration $t$. This method proceeds iteratively for $t=1,2,\dots, T$ at which point the algorithm terminates and the resulting ${\mathbf w}_T$ is the final update.  For the class of smooth non-convex loss functions, with a suitable choice of $\gamma_t$, this method converges to a local optimum at a rate $O(\frac{1}{\sqrt{N\cdot T}})$, where $N$ is the number of users. \cite{article}

\subsection{Proposed Compression Technique}
In practice the size of the gradient update vectors ${\mathbf g}_{i}^{t}$ can be very large. When training deep learning models involving millions of parameters, it could be impractical to transmit  ${\mathbf g}_{i}^{t}$ at each iteration from the client node to the central node. As a result,  compression is applied to ${\mathbf g}_{i}^{t}$ to reduce the dimensionality before transmission. A number of different approaches have been proposed in the literature. Some popular techniques include top-$k$, rand-$k$ and sketching based methods \cite{Aji2017SparseCF}\cite{NEURIPS2019_75da5036}\cite{Strom2015ScalableDD}\cite{Lin2018DeepGC}

Our proposed compression method applies a linear transformation for generating a sketch of the gradient vector as stated below:
\begin{equation}
    \label{eq:compress}
    \boldsymbol{u}^t_i = \boldsymbol{\Phi}_t \boldsymbol{g}^t_i,
\end{equation}
where $\boldsymbol{\Phi}_t \in \mathbb{R}^{s \times d}$ is a random matrix defined in the sequel. We will throughout assume that $s <  d$, so that the multiplication reduces the dimension of the gradient vector. As such, we select $\boldsymbol{\Phi}_t$ to satisfy $\mathbb{E}_{\boldsymbol{\Phi}_t}[ \boldsymbol{\Phi}_t^T \boldsymbol{\Phi}_t] = \alpha \boldsymbol{I}_d$ where $\boldsymbol{I}_d$ is a $d \times d$ identity matrix. An unbiased estimate of the gradient vector can then be computed by 
\begin{equation}
    \label{eq:Fdef}
    \hat{\boldsymbol{g}}_i^t  =  \frac{1}{\alpha} \boldsymbol{\Phi}_t^T \boldsymbol{u}_i^t. 
\end{equation}

Using this method, the encryption function~\eqref{eq:encrypt} at each client node is modified as:
\begin{equation}
    \boldsymbol{ct}^t_i = BFV.Encrypt(\boldsymbol{cpk}, \boldsymbol{u}_i^t).
\end{equation}
Notice that since the encryption operation is performed on a reduced dimensional space, it can lead to faster computation time. We will discuss this further in the experimental section.

Note the aggregated ciphertext, $\boldsymbol{ct}^t = \sum \boldsymbol{ct}_i^t$, encrypts the aggregate , $\sum_i  \boldsymbol{u}_i^t$ due to linearity of the transformation such that $\bar{\boldsymbol{u}}^t = BFV.Decrypt(s, \boldsymbol{ct}^t) = \sum {\boldsymbol{u}}_i^t $.

An unbiased estimate of the aggregate gradient can be recovered by
\begin{equation}
    \label{eq:estimate}
    \hat{\bar{\boldsymbol{g}}}^t = \frac{1}{N} \frac{1}{\alpha} \boldsymbol{\Phi}_t^T \bar{\boldsymbol{u}}^t .
\end{equation}

We note that $\boldsymbol{u}_i^t$ can be further compressed using quantization methods, such as one-bit sign quantization. In particular we can define a quantization operator:
    \begin{equation}   
    \label{eq:gs}     
    \mathcal{G}_S(x) = 
        \begin{cases}
        +1, & x \geq 0 \\ 
        -1, & x<0
        \end{cases}
    \end{equation}
and compute 
\begin{equation}
\label{eq:sign}
{\mathbf v}_i^t =  \mathcal{G}_S({\mathbf u}_i^t),
\end{equation} 
where the operator is applied element-wise. We can then reconstruct the aggregate estimate     $\hat{\bar{\boldsymbol{g}}}^t$ as in~\eqref{eq:estimate} with $\bar{\boldsymbol{u}}^t $ replaced by $\bar{\boldsymbol{v}}^t $, the aggregate of ${\mathbf v}_i^t$. We will also discuss convergence guarantees for such schemes.

%\subsubsection{Choice of Compression Matrix}
%We use a sparse subspace embedding matrix for compression\cite{DBLP:journals/corr/abs-1012-1577}. 

We now specify the ensemble of the compression matrices $\Phi_t$. We assume that at each iteration $t$, the matrix ${\mathbf \Phi}_t$ is sampled i.i.d. from a distribution $p_{{\mathbf \Phi}}(\cdot)$ which is described as follows. Each element ${\mathbf \Phi}_{t_{ij}}$ is sampled independently of all other elements and 
\begin{equation}
    \label{phi_distribution}
    \Pr({\mathbf \Phi}_{t_{ij}}=1)= 
    \Pr({\mathbf \Phi}_{t_{ij}}=-1) = p, \qquad
    \Pr({\mathbf \Phi}_{t_{ij}}=0)=1-2p.
\end{equation}
This specific distribution is chosen because it gives us a handle on its average sparsity that can significantly help reduce the computational complexity. Further for convenience, we will define $\alpha = 2 s p$ as the expected number of non-zero entries in each column, and $ r = \frac{d}{s}$ to denote the compression ratio. Note that we assume a common random seed shared at all the nodes so that such matrices can be sampled at the start of each iteration without any need for additional communication overhead. 

We note that compression methods are originally developed to reduce the communication requirements. Sparsification based methods such as top-k and rand-k \cite{Aji2017SparseCF} are not easily compatible with our framework. However, compression based on linear dimensionality reduction is a natural choice to reduce both communication and computation in our framework. \cite{NEURIPS2019_75da5036} propose sketched-SGD as a compression method based on the count sketch data structure. Sketched-SGD provides significant reduction in communication by reducing the dimensionality of gradients. We propose a similar method with a different ensemble of compression matrices that result in a handle over compression time. Additionally, we provide an extension of our method to incorporate sign quantization and provide non-trivial convergence analysis that has not been studied before.

Before stating the convergence of our proposed compression method we note the following properties that are useful in the analysis
% (the proofs are in the Appendix):
\begin{theorem}
    \label{phi_properties}
    If $\mathbf{g} \in \mathbb{R}^d$ is a $d$-dimensional vector, $\mathsf{\mathbf{\Phi}}$ a $s \times d$ random matrix sampled according to distribution in \eqref{phi_distribution}:
    \begin{align}
        \label{eq:expectation}
        \mathbb{E}_{ \mathbf{\Phi}}\big[\frac{1}{\alpha} \mathbf{\Phi}^T \mathbf{\Phi} \mathbf{g} \big] &= \mathbf{g} \nonumber \\ \mathbb{E}_{\mathbf{\Phi}}{\left \lVert \frac{1}{\alpha}\mathbf{\Phi}^T \mathbf{\Phi} \mathbf{g} \right \rVert}_2^2  &\leq (r+1+\frac{1}{\alpha}) {\left \lVert  \mathbf{g} \right \rVert}_2^2 \nonumber \\ 
        \mathbb{E}_{\mathbf{\Phi}}{\left \lVert \mathsf{\mathbf{\Phi}} \mathbf{g} \right \rVert}_1  &\geq \frac{\alpha}{\alpha r + 1} {\left \lVert  \mathbf{g} \right \rVert}_1.
    \end{align}
\end{theorem}

Note that the first property in~\eqref{eq:expectation} was already required in the system model. The second  property allows us to bound the variance of the estimate. Furthermore, last property allows us to extend this method to incorporate one bit-sign quantization. We note that $\alpha$ can be used to trade-off convergence rate and computation complexity. Larger $\alpha$ would result in a smaller variance and therefore, better convergence rate. However, a larger $\alpha$ reduces the sparsity of compression matrix and increase the computational complexity.

\subsection{Convergence Guarantees}

In order to establish the convergence of the training algorithm in the presence of compression operation we consider a modified version of the update rule in~\eqref{eq:update}. In  particular we adapt the error feedback framework in ~\cite{Stich2018SparsifiedSW} which enables us to provide convergence  guarantees for a wide variety of compression methods. Instead of applying compression to ${\mathbf g}_i^t$, user $i$ generates error compensated as $\boldsymbol{p}_i^t = \gamma \boldsymbol{g}_i^t + \boldsymbol{e}_i^t $, where $\boldsymbol{e}_i^t \in \mathbb{R}^d$ demonstrates the error term and $\gamma$ is the step-size (assumed to be constant for simplicity). 
The error term is the difference between the compressed vector and the original vector $\boldsymbol{e}^{t+1}_i = \boldsymbol{p}_i^t - \mathcal{F}(\boldsymbol{p}_i^t)$, where ${\mathcal F}(\cdot)$  denotes the overall  operator associated with our scheme:
\begin{equation}
    \label{eq:f-def2}
    \mathcal{F}(\boldsymbol{p}_i^t) = \beta(\boldsymbol{p}_i^t) \mathbf{\Phi}^T \mathbf{\Phi} \boldsymbol{p}_i^t.
\end{equation}

The compensated gradient is compressed and transmitted to central node and the update rule applied is given as,
\begin{equation}
    \boldsymbol{w}_{t+1} = \boldsymbol{w}_t - \hat{\mathbf g}_t,
    \label{eq:update2}
\end{equation}
where, $\hat{\mathbf g}_t = \frac{1}{N} \sum_{i} \mathcal{F}( {\mathbf p}_i^t)$ is the aggregate of the error compensated gradients. 

The convergence associated with the update rule~\eqref{eq:update2} can be established under very general conditions on ${\mathcal F}(\cdot)$ as stated below.\cite{Stich2018SparsifiedSW} 

\begin{definition}
    \label{compressor_definition}
    (Compression Operator) Random function $\mathcal{F}: \mathbb{R}^d \to \mathbb{R}^d$ is called a compression operator if a positive constant $\delta \leq 1 $ exists such that $\forall \mathbf{x} \in \mathbb{R}^d$ we have
    \begin{equation}
        \mathbb{E}_{\mathcal{F}}[{\left\lVert \mathcal{F}(\mathbf{x}) - \mathbf{x} \right\rVert}_2^2] \leq (1-\delta) {\left \lVert \mathbf{x} \right \rVert}_2^2,
    \end{equation}
    where expectation is taken over randomness in the operator.
\end{definition}

Recent works \cite{basu2019qsparselocalsgd} demonstrate that when error compensated gradients are used with a compression operator ${\mathcal F}$, the training error converges at a rate that scales as $O(\frac{1}{\sqrt{NT}}) + O(\frac{1}{T \delta^2})$ and thus the leading term matches the rate of SGD (without compression). In view of this it suffices to show that our proposed scheme results in a compression operator as in Definition~\ref{compressor_definition}. Towards this end, we establish the following
% (see Appendix for proof):
\begin{theorem}
    \label{no_sign_compressor}
    (Random Linear Compressor (RLC)) The operator $\mathcal{F}(\cdot)$ defined as~\eqref{eq:f-def2} is a compressor based on Definition \ref{compressor_definition} with $\beta = \frac{1}{\alpha(r + 1+1/\alpha)}$, and the  compression coefficient $\delta = \frac{1}{r+1+1/\alpha}$.
\end{theorem}

Note that Theorem~\ref{no_sign_compressor} establishes the convergence when we perform dimensionality reduction of the gradient vectors~\eqref{eq:compress}. We can also establish convergence when sign quantization is applied after dimensionality reduction~\eqref{eq:sign}.  In this case, the operator ${\mathcal F}(\cdot)$ in~\eqref{eq:f-def2} is modified as follows
\begin{align}
    \label{eq:operator_signed}
    \mathcal{F}(\boldsymbol{p}_i^t) = \beta(\boldsymbol{p}_i^t) \mathbf{\Phi}^T {\mathcal G}_s\left(\mathbf{\Phi} \boldsymbol{p}_i^t\right),
\end{align}
where ${\mathcal G}_s(\cdot)$ is the sign quantization function~\eqref{eq:gs}.
We can show the following result
% (see Appendix for proof):
\begin{theorem}
    \label{sign_compressor}
    (SRLC) The operator $\mathcal{F}(\cdot)$ defined as~\eqref{eq:operator_signed} is a compressor based on Definition \ref{compressor_definition} with $\beta(\boldsymbol{x}) = \frac{{\left \lVert \mathbf{x} \right \rVert }_1}{d(1 + \alpha)(1 + \alpha r)}$ and the following compression coefficient $\delta(\mathbf{x}) = \frac{\alpha}{(1+\alpha)(1+r\alpha)^2}\rho(\mathbf{\mathbf{x}})$, where we define $\rho(\mathbf{x})=\frac{\left \lVert \mathbf{x} \right \rVert_1^2}{d\left \lVert \mathbf{x} \right \rVert_2^2}$.
\end{theorem}

\section{Addition of New Users}
%In this section, we discuss extensions to the proposed algorithm to make it further suitable for the federated learning setup. 
%\subsection{New Users}
Our method described so far supports the presence of any subset of $k$ users to participate in successful decryption during the aggregation step. However all these users are required to have been present during the key-setup phase. In many FL applications users can join the system at any stage. While such users can still contribute in the training process by computing gradients on their available dataset, they will not be able to contribute towards the $k$ required users for successful decryption as they do not have shares of the secret key. We discuss how our proposed scheme can accommodate such users by distributing shares of the original secret key. In fact the presence of any $k$ users with secret shares of $s$ is sufficient to add a new user in the system. 

For simplicity let the original users be denoted by the set $\{0,\ldots, N-1\}$ and let the new user be denoted by $N$. Recall that the shares of $s$ at  user $i$ is denoted by $s'_i$.
Following Prop.~\ref{secret_shares_addition}, each $s'_i$ is an evaluation of the polynomial $f(x)$ at $x= x_i$ where 
\begin{equation}
    \label{eq:test1}
    f(x) =  s + \sum_{l=1}^{k-1} t_l \cdot x^l.
\end{equation}

Next we let, $\{a_i\}_{i=0}^{k-1}$ be the set of selected users to support the addition of the new user $N$. Following Theorem~\ref{shamir_reconstruction}, we are guaranteed the existence of coefficients $\{ r_{a_i}^t \}$ and $\{ r_{a_i, l}^t \}$ to reconstructed $s$ and $t_l$ respectively, i.e., 
\begin{equation}
    \label{eq:test2}
    s = \sum_{i=0}^{k-1} s'_{a_i} r_{a_i}^t \;\;\;\;  t_l = \sum_{i=0}^{k-1} s'_{a_i} r_{a_i, l} ^t.
\end{equation} 
Client $a_i$ constructs an auxiliary polynomial defined as $\hat{f}_{a_i}(x) = r_{a_i}^t s'_{a_i} + \sum_{l=1}^{k-1} r_{a_i, l}^t s'_{a_i} x^l$. Now observe that from~\eqref{eq:test1} and~\eqref{eq:test2}, we have that $f(x)  = \sum_{j=0}^{k-1} \hat{f}_{a_j}(x)$. Based on the mentioned decomposition now we can generate an evaluation of $f(\cdot)$ at $x_{N}$ using the set of selected users. Client $a_i$ would transmit $\hat{ns}_{a_i} = \hat{f}_{a_i}(x_{N})$. Then the share at the new user can be generated by simply adding these evaluations. Note that while the transmission of $\hat{ns}_{a_i}$ occurs through the central server, we employ a public key encryption system between clients to guarantee the confidentiality of the transmitted message, preventing any potential leakage to the central server. Additionally, it is evident that our system can easily accommodate scenarios involving the addition of multiple new users based on the property of Shamir's Secret Sharing.

\section{Experiments}

In this section, we experimentally evaluate the performance of the proposed scheme. Initially, we provide results comparing different secure aggregation algorithms and demonstrate the comparison between several gradient compression scheme.
\subsection{Standalone Secure Aggregation}
\paragraph{Baseline. }
We compare our proposed RSA scheme against the following baselines: TurboAggregate \cite{DBLP:journals/corr/abs-2002-04156}, Bonawitz et al. \cite{10.1145/3133956.3133982} and Adaptive Secure Aggregation scheme (ASA) described in Section~\ref{sec:statement}.

% {\em Adaptive Secure Aggregation {\bf ASA}}: ASA builds on the top of the technique of MPHE scheme \cite{Mouchet2020MultipartyHE}, which is  a way of approaching the user dropout problem by running the entire process (setup and aggregation phase) at each iteration for the set of available users. At any given iteration, the central node selects a subset of $k$ users among the available set. The users perform the setup phase followed by the decryption phase in each iteration. We note that the scheme assumes that client dropouts cannot occur between the setup phase and encryption phase in each iteration, which might not hold in practice. }}
 % We will see in the experiments that the cost of having to run a setup phase at each iteration is higher than our proposed schemes. Furthermore the scheme assumes that client dropouts cannot occur between the setup phase and encryption phase in each iteration, which might not hold in practice.
\label{sec:exper}
\begin{figure}[h]
     \centering
     \begin{subfigure}[b]{0.32\textwidth}
         \centering
         \includegraphics[scale=0.32]{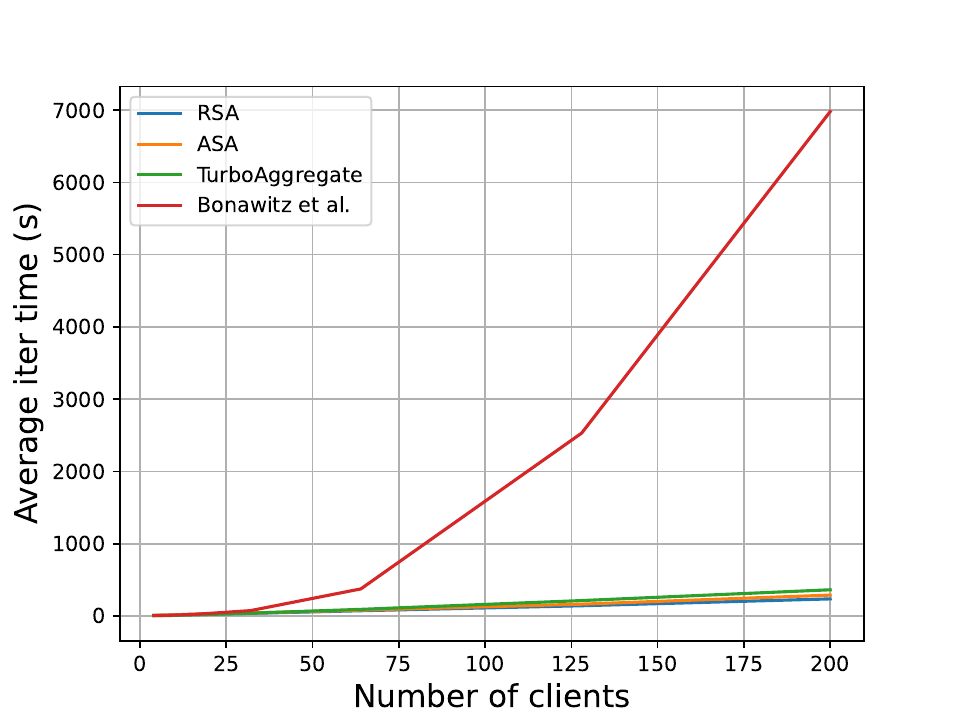}
        \caption{Comparison of all the schemes.}
        \label{}
     \end{subfigure}
     \begin{subfigure}[b]{0.32\textwidth}
        \centering
        \includegraphics[scale=0.32]{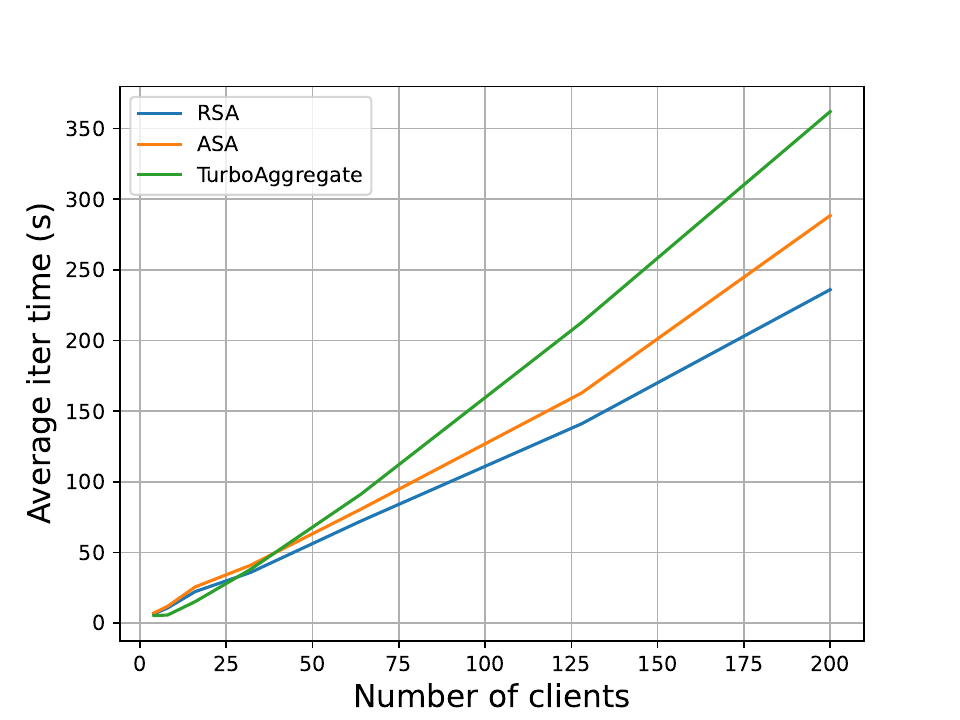}
        \caption{Comparison - exclude Bonawitz et al.}
        \label{}
    \end{subfigure}
    \caption{Total run-time of different schemes}
    \label{fig:comp_tot_runtime}
\end{figure}

\begin{figure}[h]
     \centering
     \begin{subfigure}[b]{0.32\textwidth}
         \centering
         \includegraphics[scale=0.32]{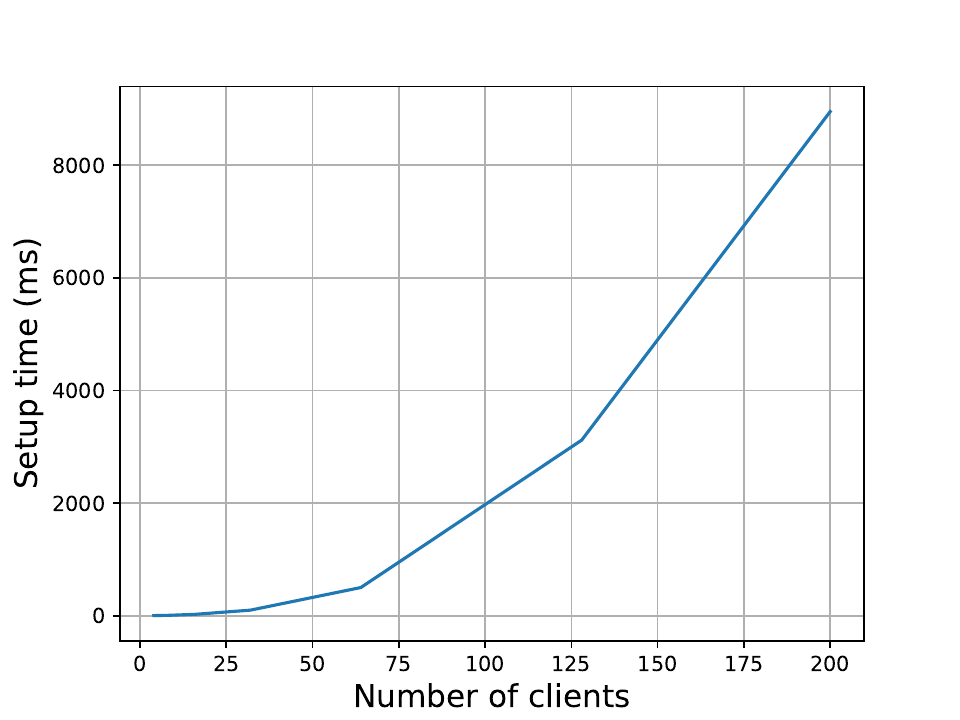}
        \caption{Setup time}
        \label{}
     \end{subfigure}\\
     \begin{subfigure}[b]{0.32\textwidth}
        \centering
        \includegraphics[scale=0.32]{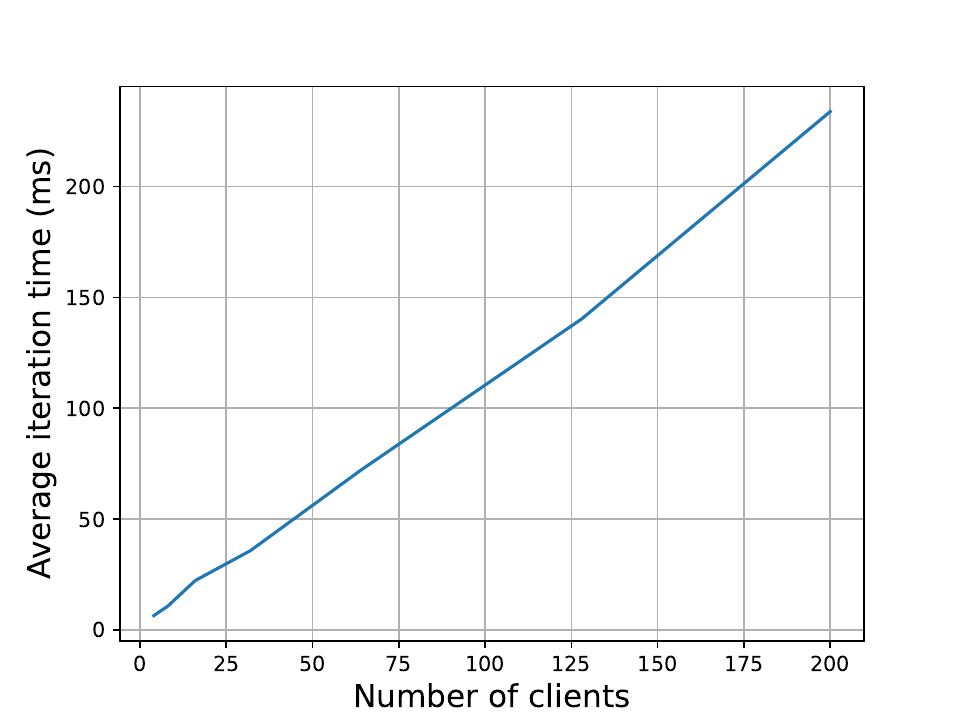}
        \caption{Average iteration time}
        \label{}
    \end{subfigure}
    \caption{Setup time and average aggregation time of RSA }
    \label{fig:RSA_breakdown}
\end{figure}

\begin{figure}[h]
    \centering
    \includegraphics[scale=0.32]{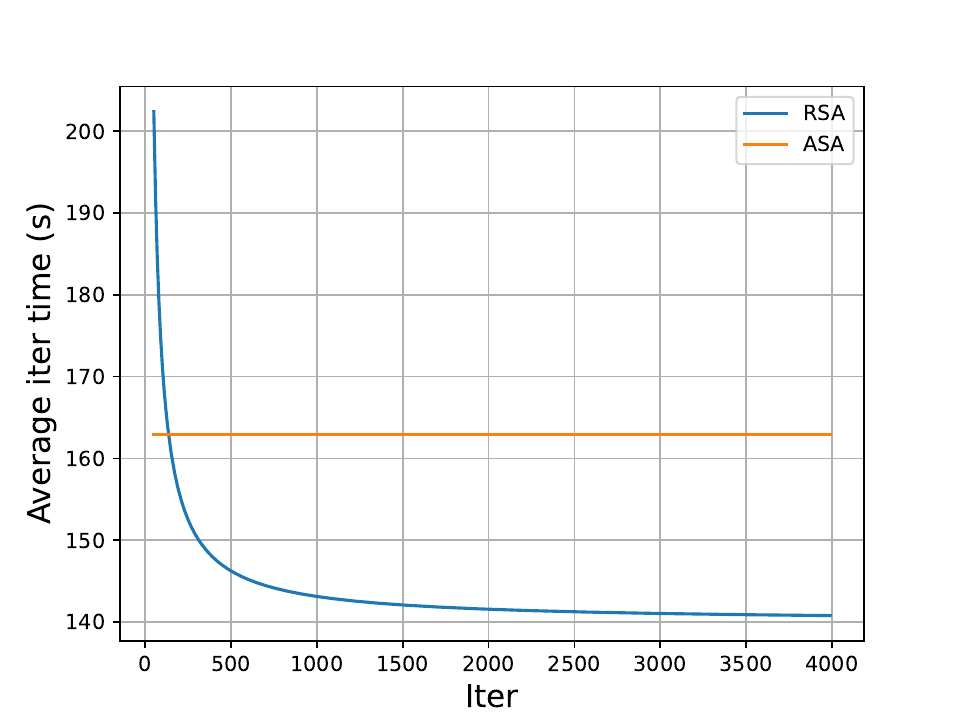}
    \caption{Effect of $T$ on average iteration time of ASA and RSA}
    \label{fig:RSA_ASA_vs_T}
\end{figure}
\paragraph{Implementation and Configurations.}
We implement all the algorithms based on Lattigo and FedML libraries \cite{lattigo} \cite{chaoyanghe2020fedml}. We implement the ASA scheme, by randomly selecting $k$ clients at each iteration, and performing the MPHE setup phase, followed by the aggregation phase. For RSA and ASA, we set the ciphertext space size $q$ to a $60$ bit prime number, the order of encryption polynomial, $n$, to $8192$, and the rest of the parameters in accordance with \cite{HomomorphicEncryptionSecurityStandard} to achieve at least $256$ bit security. Parameters of TurboAggregate are selected in accordance with \cite{DBLP:journals/corr/abs-2002-04156}. We run the experiments using the SAVI testbed (https://www.savinetwork.ca) with up to $200$ clients each holding a sequence of $T=4000$ random integer vectors of size $200,000$ and report the average iteration time as the total run-time, including the setup time, normalized by the number of iterations. We set the dropout rate to be $25\%$ such that $k=\frac{3}{4}N$ for all the algorithms.  We note that the MPHE setup phase does not include the Shamir's secret sharing step. 

\paragraph{Performance.}
Figure~\ref{fig:comp_tot_runtime} compares the average run-time of different schemes. Even though the method proposed by Bonawitz et al. demonstrates a better performance for smaller values of $N$, the performance drops dramatically for larger number of clients due to the quadratic complexity. Furthermore, we observe that both RSA and ASA demonstrate a linear behavior. However, RSA consistently delivers better run-time compared to ASA because of better round efficiency. We note that this is consistent with the results of Section~\ref{complexity_analysis}, as $T >> N$ in our setup.
Further comparing ASA and RSA, figure~\ref{fig:RSA_ASA_vs_T} plots the average iteration time of both schemes for different values of $T$ when $N=128$ and $k=96$. We observe that for smaller values of $T$, the setup time becomes the dominating factor and adversely affects the performance of RSA. However, increasing $T$, RSA achieves a better performance compared to ASA.

Further analyzing the behavior of RSA, Figure~\ref{fig:RSA_breakdown} demonstrates the behavior of setup time and average aggregation time of RSA as the number of clients changes. We observe the quadratic behavior of setup time and linear behavior of average iteration time.

\subsection{Compression}
In this section we initially provide results on the performance of the proposed compression operators. Then, we experimentally evaluate how compression can improve the run-time of secure aggregation in federate learning.
\subsubsection{Compression Operators}
\begin{table}[h]
    \centering
    \begin{tabular}{c|c|c}
        Method & Test Accuracy & Compression Time \\
        \hline
        Baseline SGD & $95.48\%$ & $0$\\ 
        \hline
        sketched-SGD, x1000 & $94.95\%$ & $19.67s$\\ 
        SRLC, x1000, $\alpha=1$ & $95.32\%$ & $20.6s$\\
        SRLC, x1000, $\alpha=0.01$ & $95.28\%$ & $0.25s$\\
        
        \hline
        sketched-SGD, x2000 & $94.48\%$ & $19.43s$ \\
        SRLC, x2000 , $\alpha=1$ & $95.26\%$ & $19.64s$\\
        SRLC, x2000 , $\alpha=0.01$ & $95.20\%$ & $0.24s$ 
        
    \end{tabular}
    
    \caption{Test accuracy and per epoch compression time for training with different methods}
    \label{tab:sgd_sketchedsgd_srlc}
\end{table}
% \begin{figure}[h]
%     \centering
%     \includegraphics[scale=0.45]{test_acc_comparison.pdf}
%     \caption{Convergence behavior of training by different methods}  
%     \label{fig:convergence_comparison}
% \end{figure}

We compare the performance of our proposed compression scheme called Sign-quantization Random Linear Compressor (SRLC), with uncompressed SGD and sketched-SGD using Tensorflow\cite{tensorflow2015-whitepaper}. We select the parameters of sketched-SGD in accordance with \cite{NEURIPS2019_75da5036}. We train EfficientNetB4 \cite{DBLP:journals/corr/abs-1905-11946} with about $17.5$M parameters over CIFAR-10 \cite{Krizhevsky2009LearningML} using $8$ workers each having a batch size of $32$, each worker is simulated on a GPU. We train all models for $160$ epochs, each epoch consisting of $196$ iterations, using an initial learning rate of $0.1$. We decrease the learning rate by a factor of $10$ at epochs $80$ and $120$. 
Table \ref{tab:sgd_sketchedsgd_srlc} summarizes the final test accuracy and per epoch compression time achieved by different training methods and levels of compression. We use $\alpha \in \{0.01, 1\}$ for training models using SRLC. Considering $\alpha=1$, SRLC and sketched-SGD have the same computational complexity. We observe that under the same compression time and communication requirements, SRLC achieves a better test accuracy compared to sketched-SGD. Interestingly in our proposed scheme, selecting $\alpha=0.01$ significantly reduces compression time without sacrificing the accuracy. 

% Figure \ref{fig:convergence_comparison} further depicts the convergence behavior of the algorithms for x1000 compression level, we use $\alpha=0.01$ for SRLC. For both of the methods, training curve closely follows the one for uncompressed SGD. However, SRLC with $\alpha=0.01$ has a considerably lower computational complexity than sketched-SGD. 
\subsubsection{Compressed Secure Federated Learning}
\begin{table}[h]
    \centering
    \begin{tabular}{c|c|c}
        Method & Iteration Time (s) & Accuracy \\ \hline
        SGD & $2.1$  & $96.3\%$   \\ \hline
        $r=1$ & $12.3$ & $95.4\%$ \\
        $r=5$ & $6.7$  & $95.4\%$ \\
        $r=10$ & $4.4$  & $95.2\%$ \\
        $r=50$ & $3.3$  & $94.9\%$ \\
        $r=75$ & $2.9$  & $94.5\%$ \\
    \end{tabular}
    \caption{Effect of compression on performance of the scheme}
    \label{exp_compression_table}
\end{table}

In this part, we use the same parameters as of section \ref{sec:exper} for aggregation. We train a convolutional neural network with about $600$K parameters and MNIST dataset\cite{lecun2010mnist}. We perform training using $8$ clients and a batch size of $256$ per client. We simulate each worker on a different machine with 8 vCPUs and no access to GPU. We set the learning rate to be $0.1$ throughout the learning process and train for $300$ iterations. The dropout rate is set to be $50\%$ such that $k=4$, meaning that at least half of the network should be active for decryption. In order to convert floating-point gradients to integers, we multiply each element by a scaling factor of $10^3$ and round it to the nearest integer. We further clip the gradients with $10^3$ based on the magnitude in order to avoid overflow. We train using RLC compressor with $\alpha=0.1$ and utilize the error feedback method. We observe that the run-time for the secure algorithm is about $6$ times longer than the plain SGD. Examining the effect of compression, we train with different values of compression ratio $r \in \{ 5, 10, 50, 75 \}$. We note that any compression ratio larger than $75$ would not yield any additional advantage, as we need to perform encryption at least once. Table \ref{exp_compression_table} exhibits the effect of compression on the performance. As we can see, with an increase in the compression ratio, the average iteration time drops because the number of sequential encryption operations is reduced by a factor of $r$. We also observe a negligible drop in final accuracy with an increase in compression ratio. We note that much higher compression ratios are possible as shown in previous sections. Therefore, this method is applicable for much larger models without any significant drop in terms of accuracy. We further note that lower accuracy compared to SGD is partly due to the additional noise imposed in the quantization process. 

% Furthermore, Figure \ref{exp_compression_fig} demonstrates the convergence curve of training with different levels of compression. As we can, a larger compression ratio results in a slower convergence. However, the final result achieved by all the schemes are similar.

% \section{Conclusion}
% \label{sec:concl}

% We propose the use of multiparty homomorphic encryption (MPHE) for secure aggregation in federated learning. We first extend MPHE schemes by introducing a secret-sharing step during the setup phase of MPHE, so that  decryption of the aggregate gradient vector can succeed even when a subset of client nodes are able to participate. This is in contrast to classical MPHE schemes that require the participation of all the users. We establish conditions on the MPHE parameters that guarantee correctness and security. We also discuss how our proposed scheme can be extended to support users who are not present during the setup phasae. We  propose a technique for compressing the gradient vectors using a linear trnsform that is compatible with the proposed MPHE scheme. We study the associated convergence properties when error feedback is applied. We finally demonstrate a practical implementation based on the Lattigo library and CKKS encryption scheme. Our experimental results demonstrate that while the computational overhead associated with encryption can be significant, it can be substantially mitigated when gradient compression is applied before encryption. 
\section{Conclusion}

In this work, we propose a cryptographic secure aggregation protocol called robust secure aggregation (RSA) for gradient aggregation in Federated Learning by novelly combining shamir's secret sharing and MPHE scheme. RSA addresses two key issue in Federated Learning. It not only supports clients dropped-out when a fraction of clients involved in each training iteration and also can be extended to accommodate new client joining. Additionally, RSA shows superiority in communication complexity over prior works, especially for long-term training. We also propose a compression scheme named SRLC for gradient vectors at each client node to reduce the communication overhead introduced by cryptographic tool.

\section{Acknowledgements}
We wish to thank and acknowledge the help provided by professor Shweta Agrawal in validating the security analysis provided in this work.

\bibliographystyle{ieeetr}
\bibliography{double_column}

\appendices

\begin{tight}
\section{Compression Proofs}

\textbf{Proof of theorem \ref{phi_properties}}  We start by proving the first property. Doing so, we will prove that 
$\mathbb{E}_{\mathsf{\mathbf{\Phi}}}\big[ \mathsf{\mathbf{\Phi}}^T \mathsf{\mathbf{\Phi}}\big] = \alpha \mathbf{I}_d$. We do this by expanding $\mathsf{\mathbf{\Phi}}^T \mathsf{\mathbf{\Phi}}$, 
\begin{align}
    \mathbb{E}_{\mathsf{\mathbf{\Phi}}}\big[ (\mathsf{\mathbf{\Phi}}^T \mathsf{\mathbf{\Phi}})_{ij}\big] = \mathbb{E}_{\mathsf{\mathbf{\Phi}}}\big[ \sum_k \mathsf{\mathbf{\Phi}}_{ki} \mathsf{\mathbf{\Phi}}_{kj} \big]
    = 2sp \mathbf{I}_d
\end{align}
We prove the second property by showing that $\mathbb{E}_{\mathsf{\mathbf{\Phi}}}\big[ (\mathsf{\mathbf{\Phi}}^T \mathsf{\mathbf{\Phi}})^2\big] = \big(4p^2s(s-1)+2sp+4p^2s(d-1)\big) \mathbf{I}_d$
. We show this by expanding $(\mathsf{\mathbf{\Phi}}^T \mathsf{\mathbf{\Phi}})^2$,
\begin{align}
    \mathbb{E}_{\mathsf{\mathbf{\Phi}}}\big[ (\mathsf{\mathbf{\Phi}}^T \mathsf{\mathbf{\Phi}})^2_{ij}\big] &= \mathbb{E}_{\mathsf{\mathbf{\Phi}}}\big[ \sum_{k,l,t} \mathsf{\mathbf{\Phi}}_{li}\mathsf{\mathbf{\Phi}}_{lk}\mathsf{\mathbf{\Phi}}_{tk}\mathsf{\mathbf{\Phi}}_{tj}\big] \nonumber \\
    &= \sum_{l = t, i=j=k} 2p + \sum_{l= t, i=j\ne k}4p^2 + \sum_{l\ne t, i=j=k} 4p^2 \nonumber \\
    &= 4p^2s(s-1)+2sp+4p^2s(d-1)
\end{align}
We complete the proof by noting $4p^2s(s-1)+2sp+4p^2s(d-1) \leq 4p^2 s^2 + 2sp + 4p^2sd \leq \alpha^2(r+1+1/\alpha)$.
Proving the third property, first we show that 
\begin{align}
    \label{lemma_for_norm_1}
    \sum_{\mathbf{A}_i \in \{-1,+1\}^d} |\mathbf{A}_i^T \mathbf{g}| \geq \frac{2^d}{d} \norm{\mathbf{g}}_1 
\end{align}
We show this by induction on $d$. The case of $d=1$ is trivially satisfied with equality. Now we assume that the relation is true for $d$. Let $\mathbf{g}$ denote a $d+1$ dimensional vector and $\mathbf{g}'$ be a $d$ dimensional vector identical to $\mathbf{g}$ with $j$-th element removed, so that, $\mathbf{g}' = (g_1, ..., g_{j-1}, g_{j+1}, ..., g_{d+1})$. 
\begin{align}
    \label{induction_rule}
    \sum_{\mathbf{A}_i \in \{-1,+1\}^{d+1}} |\mathbf{A}_i^T \mathbf{g}| &\geq 
    \sum_{\mathbf{A}_i \in \{-1,+1\}^{d}} |g_j + \mathbf{A}_i^T \mathbf{g}'|  \nonumber \\
    &+ |-g_j + \mathbf{A}_i^T \mathbf{g}'|  \nonumber \\
    &\geq \sum_{\mathbf{A}_i \in \{-1,+1\}^{d}} |g_j| + |\mathbf{A}_i^T \mathbf{g}'| \nonumber \\
    &\geq 2^d|g_j|  + \frac{2^d}{d} \norm{\mathbf{g}'}_1  
\end{align}
Where we use the following inequality $|a+b|+|a-b| \geq |a| + |b|$
The result at \eqref{induction_rule} holds for any $j$. We complete the proof by summing all such inequalities we have,
\begin{align}
   (d+1) \sum_{\mathbf{A}_i \in \{-1,+1\}^{d+1}} |\mathbf{A}_i^T\mathbf{g}| \geq (2^d + d\frac{2^d}{d}) \norm{\mathbf{g}}_1 
   = 2^{d+1} \norm{\mathbf{g}}_1 
\end{align}
Now we can state the proof for third property. Let $\boldsymbol{\phi}_i$ denote the $i$-th row of $\mathsf{\mathbf{\Phi}}$. We have by linearity of expectation, $\mathbb{E}_{\mathsf{\mathbf{\Phi}}}[\norm{\mathsf{\mathbf{\Phi}} \mathbf{g}}_1 ] = \sum_i \mathbb{E}_{\boldsymbol{\phi}_i}[|\boldsymbol{\phi}_i^T \mathbf{g}|]$. Let $\Omega_k$ denote the set of all possible $\boldsymbol{\phi}_i$ that have $k$ non-zero elements. Based on \eqref{lemma_for_norm_1} we can show that, $\sum_{\boldsymbol{\phi}_i \in \Omega_k} |\boldsymbol{\phi}_i^T \mathbf{g}| \geq \frac{2^k}{k} {d-1 \choose k-1} \norm{\mathbf{g}}_1$. Using this result we expand the expectation as
\begin{align}
    \mathbb{E}_{\mathsf{\mathbf{\Phi}}}[\norm{\mathsf{\mathbf{\Phi}} \mathbf{g}}_1 ] = \sum_i \sum_{k=1}^d \sum_{\mathbf{\phi}_i \in \Omega_k} \mathbb{P}[\boldsymbol{\phi}_i] \big|\boldsymbol{\phi}_i^T \mathbf{g}\big| \nonumber \\
    \geq \sum_{i,k=1}^d 
    \frac{2^k}{k} {d-1 \choose k-1} p^k (1-2p)^{d-k} \norm{\mathbf{g}}_1 \nonumber \\ \nonumber
    \geq \sum_{i, k'=0}^{d-1}
    \frac{1}{k'+1} {d-1 \choose k'} (2p)^{k'+1} (1-2p)^{d-k'-1}\norm{\mathbf{g}}_1 \\
    \geq 2p \sum_{i,k'=0}^{d-1}
    \frac{1}{k'+1} {d-1 \choose k'} (2p)^{k'} (1-2p)^{d-k'-1}\norm{\mathbf{g}}_1
\end{align}
Let $X = X_1 + X_2 + ... + X_{d-1}$ be a random variable denoting the sum of $d-1$ i.i.d. $\text{Bernoulli}(2p)$ random variables. We can show that, $\mathbb{E}[\frac{1}{X+1}] = \sum_{k'=0}^{d-1}\frac{1}{k'+1} {d-1 \choose k'} (2p)^{k'} (1-2p)^{d-k'-1} $
.Then by Jensen's inequality we have, $\mathbb{E}[\frac{1}{X+1}] \geq \frac{1}{\mathbb{E}[X]+1} = \frac{1}{2p(d-1)+1} \geq \frac{1}{2pd+1}$. Plugging in the result completes the proof,
\begin{align}
    \mathbb{E}_{\mathsf{\mathbf{\Phi}}}[\norm{\mathsf{\mathbf{\Phi}} \mathbf{g}}_1 ] &\geq \sum_i 2p\frac{1}{2pd+1}\norm{\mathbf{g}}_1 \nonumber \\
    &= \frac{2sp}{2pd+1}
    \norm{\mathbf{g}}_1 = \frac{\alpha}{\alpha r+1} \norm{\mathbf{g}}_1
\end{align}
Which completes our proof.

\textbf{Proof of theorem \ref{no_sign_compressor}} Following the definition of a compressor and $r' = r+1+1/\alpha$,
\begin{align}
    \mathbb{E}_{{\mathsf{\mathbf{\Phi}}}}[\norm{\frac{1}{\alpha r'}{\mathsf{\mathbf{\Phi}}}^T {\mathsf{\mathbf{\Phi}}} \mathbf{x} - \mathbf{x}}_2^2] &= \mathbb{E}_{{\mathsf{\mathbf{\Phi}}}}[\frac{1}{\alpha^2r'^2} \norm{{\mathsf{\mathbf{\Phi}}}^T {\mathsf{\mathbf{\Phi}}}\mathbf{x}}_2^2] \nonumber \\
    &+ \norm{\mathbf{x}}_2^2 
     - 2\mathbb{E}_{{\mathsf{\mathbf{\Phi}}}}[\frac{1}{\alpha r'} \mathbf{x}^T{\mathsf{\mathbf{\Phi}}}^T {\mathsf{\mathbf{\Phi}}}\mathbf{x}] \nonumber \\
    &\leq \frac{1}{r'}\norm{x}_2^2 + \norm{x}_2^2 - \frac{2}{r'}\norm{x}_2^2 \nonumber \\ 
    &\leq \norm{x}_2^2[1 - \frac{1}{r'}]
\end{align}
Therefore we conclude that $\delta = \frac{1}{r'} = \frac{1}{r+1+1/\alpha}$.

\textbf{Proof of theorem \ref{sign_compressor}} Based on the definition of a compressor, 
\begin{align}
    &\mathbb{E}_{\mathsf{\mathbf{\Phi}}}[\norm{\frac{{\left \lVert \mathbf{x} \right \rVert }_1}{d(1 + \alpha)(1 + \alpha r)} \mathsf{\mathbf{\Phi}}^T \mathcal{G}_S(\mathsf{\mathbf{\Phi}} \mathbf{x}) - \mathbf{x}}_2^2] = \nonumber \\  &\mathbb{E}_{\mathsf{\mathbf{\Phi}}}\left[\norm{\frac{\norm{\mathbf{x}}_1}{d (1+\alpha)(1+\alpha r)}\mathsf{\mathbf{\Phi}}^T \mathcal{G}_S(\mathsf{\mathbf{\Phi}} \mathbf{x}) }_2^2\right] \nonumber \\
    &+ \norm{\mathbf{x}}_2^2 
      - 2
    \mathbb{E}_{\mathsf{\mathbf{\Phi}}}\left[\frac{\norm{\mathbf{x}}_1}{d (1+\alpha)(1+\alpha r)} \mathbf{x}^T\mathsf{\mathbf{\Phi}}^T \mathcal{G}_S(\mathsf{\mathbf{\Phi}} \mathbf{x})\right]
\end{align}
We proceed by upper bounding $
    \mathbb{E}_{\mathsf{\mathbf{\Phi}}}[\norm{\mathsf{\mathbf{\Phi}}^T \mathcal{G}_S(\mathsf{\mathbf{\Phi}} \mathbf{x}) }_2^2]$
    by calculating each of its elements,
\begin{align}
    \mathbb{E}_{\mathsf{\mathbf{\Phi}}}\left[(\mathsf{\mathbf{\Phi}}^T \mathcal{G}_S(\mathsf{\mathbf{\Phi} \mathbf{x}}))_i^2\right] &= \mathbb{E}_{\mathsf{\mathbf{\Phi}}} \left[ \left(\sum_j {\mathsf{\mathbf{\Phi}}}_{ji} \mathcal{G}_S(\mathsf{\mathbf{\Phi} \mathbf{x}})_j\right)^2\right] \nonumber \\
    &= \mathbb{E}_{\mathsf{\mathbf{\Phi}}}\left[\sum_j {\mathsf{\mathbf{\Phi}}}_{ji}^2\right] \nonumber \\
    &+ \mathbb{E}_{\mathsf{\mathbf{\Phi}}}\left[\sum_{l \ne k} {\mathsf{\mathbf{\Phi}}}_{ki}{\mathsf{\mathbf{\Phi}}}_{li}\mathcal{G}_S(\mathsf{\mathbf{\Phi} \mathbf{x}})_k\mathcal{G}_S(\mathsf{\mathbf{\Phi} \mathbf{x}})_l\right] \nonumber \\
    &\leq \alpha + \mathbb{E}_{\mathsf{\mathbf{\Phi}}}\left[\sum_{l \ne k} \left| {\mathsf{\mathbf{\Phi}}}_{ki}{\mathsf{\mathbf{\Phi}}}_{li}  \right|\right]  \nonumber \\
    &\leq \alpha + s(s-1)4p^2 
    \leq \alpha^2 + \alpha
\end{align}
Therefore, we conclude that $\mathbb{E}_{\mathsf{\mathbf{\Phi}}}\left[\norm{\mathsf{\mathbf{\Phi}}^T \mathcal{G}_S(\mathsf{\mathbf{\Phi}} \mathbf{x}) }_2^2\right] \leq d\alpha(\alpha + 1)$. Now we bound $\mathbb{E}_{\mathsf{\mathbf{\Phi}}}[\mathbf{x}^T\mathsf{\mathbf{\Phi}}^T \mathcal{G}_S(\mathsf{\mathbf{\Phi}} \mathbf{x})]$ using the third property of compression matrix,
\begin{align}
    \mathbb{E}_{\mathsf{\mathbf{\Phi}}}\left[\mathbf{x}^T\mathsf{\mathbf{\Phi}}^T \mathcal{G}_S(\mathsf{\mathbf{\Phi}} \mathbf{x})\right] = 
    \mathbb{E}_{\mathsf{\mathbf{\Phi}}}\left[\norm{\mathsf{\mathbf{\Phi}} \mathbf{x}}_1\right] \geq \frac{\alpha}{1+\alpha r} \norm{\mathbf{x}}_1
\end{align}
Now applying the bounds to complete the proof.
\begin{align}
    \mathbb{E}_{\mathsf{\mathbf{\Phi}}}\left[\norm{\frac{{\left \lVert \mathbf{x} \right \rVert }_1}{d(1 + \alpha)(1 + \alpha r)} \mathsf{\mathbf{\Phi}}^T \mathcal{G}_S(\mathsf{\mathbf{\Phi}} \mathbf{x}) - \mathbf{x}}_2^2\right] \nonumber \\
    \leq \frac{\norm{\mathbf{x}}_1^2 d \alpha (1+\alpha)}{d^2 (1+\alpha)^2 (1+\alpha r)^2}
     + \norm{\mathbf{x}}_2^2 
     -2\frac{\alpha \norm{\mathbf{x}}_1^2}{d(1+\alpha)(1+r \alpha)^2} \nonumber \\
    \leq  \norm{\mathbf{x}}_2^2 - \frac{\alpha \norm{\mathbf{x}}_1^2}{d(1+\alpha)(1+r \alpha)^2} \nonumber \\
    \leq  \norm{\mathbf{x}}_2^2 \left[1 - \frac{\alpha \norm{\mathbf{x}}_1^2}{d\norm{\mathbf{x}}_2^2(1+\alpha)(1+r \alpha)^2}\right] \nonumber \\
    \leq
    \norm{\mathbf{x}}_2^2 \left [1 - \frac{\alpha}{(1+\alpha)(1+r \alpha)^2}\rho(\mathbf{x})\right ]
\end{align}
Which completes the proof.

\end{tight}
\end{document}